\documentclass[12pt,onecolumn,draftclsnofoot,journal]{IEEEtran}

\usepackage{amsthm}
\usepackage{amssymb}
\usepackage{latexsym}
\usepackage{amsfonts}
\usepackage{amsbsy}
\usepackage{amsmath,amssymb}
\usepackage{times}
\usepackage{graphicx}
\usepackage{enumerate}
\usepackage[usenames]{color}
\usepackage[dvips]{pstcol}
\usepackage{epstopdf}
\usepackage{subfig}
\usepackage{bm}
\usepackage{booktabs}
\usepackage{cite}
\usepackage{color}
\usepackage{xcolor}
\usepackage{algpseudocode}
\usepackage{setspace}
\usepackage{bm}
\usepackage{tabularx}
\usepackage{amstext,epsfig,graphicx}
\usepackage{psfrag}
\usepackage{cite}
\usepackage{footnote}
\usepackage{siunitx}
\usepackage{mathrsfs}
\usepackage{booktabs}
\usepackage{multirow}
\usepackage{stfloats}
\usepackage{algorithm}
\usepackage{algorithmicx}
\usepackage{amsmath}

\begin{document}

%
\title{Channel Tracking and Prediction for IRS-aided Wireless Communications}
%
%
%

\author{Yi Wei, Ming-Min Zhao, An Liu, and Min-Jian Zhao \vspace{-0.7cm}
\thanks{The authors are with the College of Information Science and Electronic Engineering, Zhejiang University, Hangzhou 310027, China
(email: \{21731133, zmmblack, anliu, mjzhao\}@zju.edu.cn).}
}
\IEEEpeerreviewmaketitle

\maketitle

\vspace{-1.5em}

\begin{abstract}
For intelligent reflecting surface (IRS)-aided wireless communications, channel estimation is essential and usually requires  excessive channel training overhead when the number of IRS reflecting elements is large.
The acquisition of accurate  channel state information (CSI) becomes more challenging when the channel is not quasi-static  due to the mobility of the transmitter and/or receiver.
In this work, we study an IRS-aided wireless communication system with a time-varying channel model and propose an innovative two-stage transmission protocol.
In the first stage, we send pilot symbols and track the direct/reflected channels based on the received signal, and then data signals are transmitted. In the second stage, instead of sending pilot symbols first, we directly  predict the direct/reflected channels and all the time slots are used for data transmission.
Based on the proposed transmission protocol, we propose a two-stage channel tracking and prediction (2SCTP) scheme to obtain the direct and reflected channels with low channel training overhead, which is achieved by exploiting the temporal correlation of the time-varying channels.
Specifically, we first consider a special case where the IRS-access point (AP) channel is assumed to be static, for which a Kalman filter (KF)-based algorithm and  a long short-term memory (LSTM)-based neural network are proposed for channel tracking and prediction, respectively. Then, for the more
general case where the IRS-AP, user-IRS and user-AP channels are all assumed to be time-varying,
we present a generalized KF (GKF)-based channel tracking algorithm, where proper approximations are employed to handle the underlying non-Gaussian random variables.
Numerical simulations are provided to
verify the effectiveness of our proposed transmission protocol and channel tracking/prediction algorithms as  compared to existing ones.
\end{abstract}

\vspace{-0cm}

\newtheorem{remark}{\bf Remark}
\newtheorem{theorem}{\bf Theorem}
\vspace{-0cm}
\linespread{1.5}

\vspace{-0.5cm}\section{Introduction}
As a promising cost-effective technology for enhancing the spectral and energy efficiency of  wireless communication systems,
 intelligent reflecting surface (IRS), also known  as reconfigurable intelligent surface (RIS), has drawn significant attention in both academy and industry \cite{9326394}. The IRS is composed of a massive number of low-cost passive reflecting elements, which can be smartly controlled to dynamically configurate  the wireless communication environment for transmission enhancement and interference suppression. Due to its passive nature, IRS requires lower hardware cost and energy consumption  as compared to the traditional active relays.   As such, IRS can be densely deployed  in  wireless communication systems to flexibly reconfigurate the propagation environment, achieving improved communication  capacity and reliability \cite{9195133}.

Several current research activities focus
on how to implement IRSs \cite{Tan2018,Subrt2012}, as well as the potentials and challenges of IRS-aided wireless communications \cite{9374975,Gong2019}. In most of the existing works, the acquisition of channel state information (CSI) is essential since the performance gain provided by IRS is heavily dependent on the channel estimation  accuracy.
However, since IRS is passive in general and can neither
send nor receive pilot symbols,
the IRS-user  and access point (AP)-IRS channels cannot be
estimated separately. Therefore,
the channel estimation problem in IRS-aided wireless communication systems is much more challenging
as compared to those in conventional systems without IRS.
Fortunately, the knowledge of the cascaded user-IRS-AP channel, also known as the \emph{reflected channel},  is sufficient for signal detection and beamforming design \cite{9115725}.
As such, most of the existing channel estimation related contributions in the IRS literature focused on the reflected channel estimation problem \cite{Mishra2019,Jensen2019,arXiv1912,wei2021channel,Jie2019,9370097,He2019,YuboWan2020,9198125}. Specifically,
in \cite{Mishra2019}, the authors proposed to estimate the cascaded channel coefficients  one-by-one by switching only one IRS reflecting element on at each time. In \cite{Jensen2019},
 a discrete Fourier transform (DFT) based  IRS phase-shift matrix (also named as \emph{reflection pattern}) was proposed, where  all IRS reflecting elements are designed to be active and the estimation variance is reduced as compared to  the scheme in  \cite{Mishra2019}.
Moreover, the works \cite{arXiv1912} and \cite{wei2021channel}  focused on the channel estimation problem in IRS-aided multi-user MISO systems. Both of them exploited the
 correlation among the reflected channels to reduce the channel training overhead, and the latter further improved the scheme in  \cite{arXiv1912} by alleviating the negative effects caused by error propagation. In addition, several works designed channel estimation algorithms based on some special properties of the channels \cite{Jie2019,9370097,He2019}. In particular, the works
 \cite{Jie2019} and \cite{9370097}  formulated the channel estimation problem as a sparse channel matrix recovery problem using the compressive sensing (CS) technique. The work \cite{He2019} exploited the low-rank structure of the channels in massive multi-input multi-output (MIMO) systems and formulated  the reflected  channel estimation problem  as a combined sparse matrix factorization problem.
Furthermore, when the full instantaneous  CSI is not available, the statistical CSI can be  utilized to design advanced active and passive beamforming algorithms, which usually incurs much less channel estimation overhead \cite{YuboWan2020,9198125}.

In most of the aforementioned works, the quasi-static channel model is assumed, i.e., the channels are assumed to be approximately constant over a relatively long coherence time, such that accurate estimation of the instantaneous CSI is possible. However, in practice,
when the users are  with mobility, the channel coefficients  are more likely
to vary and be temporally correlated, which can be utilized to further  reduce the channel training overhead. For such time-varying channels, the existing channel estimation algorithms are not efficient  in general and may even be inapplicable.
 Instead, efficient  channel tracking methods are usually required  to be designed to obtain the time-varying CSI.
Channel tracking has been widely studied in the  literature for traditional wireless communication systems without IRS \cite{995063,8007240,8630098}. Specifically,
the work \cite{995063} studied the channel
tracking and equalization problem for
time-varying frequency-selective MIMO channels by  approximating  the MIMO channel variation using  a low-order autoregressive model and  tracking the approximated channel
via a Kalman filter (KF). In orthogonal frequency division multiplexing (OFDM) systems, when time-varying frequency selective channels are considered, the work \cite{8007240} proposed to first successively track the delay-subspace by KF and then track the channel impulse response.
The authors in \cite{8630098} considered an FDD massive MIMO system with limited scattering around the  base station (BS) and
 proposed a two-dimensional (2D) Markov model to capture the 2D dynamic sparsity of massive MIMO channels. An efficient message passing algorithm was derived  to recursively track the dynamic channel.
For IRS-aided communication systems, the work \cite{9366786} considered a time-varying channel model and
 proposed to track the downlink direct and reflected channel using two independent KFs. It was assumed in \cite{9366786} that  the reflected channel, i.e., the cascaded user-IRS-AP channel, follows the Gauss-Markov model. However, this assumption may not hold in practice. Although the individual user-IRS, IRS-AP and user-AP channels can be well approximately as Gauss-Markov models as reported in many related works, the cascaded user-IRS-AP channel is no longer Gaussian distributed.

{  In this work,
we advance the abovementioned works by studying the channel tracking and prediction (CTP) problem in an IRS-aided wireless communication system and considering a more general time-varying channel model. Specifically, the
 user-IRS, IRS-AP  and user-AP
 channels are assumed to be independent with each other and all follow
 stationary steady-state Gauss-Markov processes. A two-stage transmission protocol is proposed  which is able to reduce  the channel training overhead by exploiting the temporal correlation of the channels.
 Based on the proposed transmission protocol, we design an innovative two-stage CTP (2SCTP) scheme, where
a KF-based channel tracking algorithm and a deep learning (DL)-enabled channel prediction
algorithm are integrated for performance enhancement.
The main contributions of this work are summarized as follows:}
\begin{itemize}
\item {  First,
we propose a two-stage
 transmission protocol and a 2SCTP scheme.
%
   Specifically, the first stage contains two phases, in the channel training phase, the user sends pilot symbols  and the AP tracks the channel based on the received pilot signals, while in  the data transmitting phase,  data signals are sent at the AP and the user aims to recover the data based on the tracked channel obtained in the channel training phase.
 In the second stage,  the AP
 obtains the channel coefficients via prediction and all time slots in this stage are allocated for data transmission. Adopting the proposed two-stage transmission protocol is able to reduce the channel training overhead significantly, especially in the second stage.
 }

\item {  Second, for the special case where the IRS-AP channel is assumed to be approximately constant, we propose a KF-based algorithm for channel tracking and a long short-term memory (LSTM)-based neural network for channel prediction.\footnote{Considering this special case is meaningful since the locations of the  IRS and AP are usually fixed in practice and thus the IRS-AP channel can remain constant for a relatively long time interval.} Specifically,
the  {KF-based channel tracking algorithm} is derived by modeling the channel varying model as a linear state equation and using historical channel information
(i.e., the covariance matrix of the previously estimated channels). It is able to
 recover the high-dimensional channel vector from received pilot signals (i.e., observations), whose dimension is much lower.
  In the second stage, the LSTM-based neural network, {  namely the \emph{observation (OB)-LSTM}}, is designed to predict some imaginary observations, instead of directly predicting the channels, and then the KF-based algorithm is employed again to predict the channel coefficients based on these imaginary observations.}
  {  Note that predicting the imaginary observations and then employing the KF-based algorithm is more efficient than directly predicting the channels using LSTM  since the dimensions of the channel vectors are usually much larger than the observations, especially when the number of reflecting elements is large, which makes directly predicting the channel vectors very difficult. }


\item Third, we study the general case where  the user-IRS, IRS-AP  and user-AP channels are all assumed to be  time-varying.
    In this case,  the KF-based channel tracking algorithm cannot be directly applied, since not all the random variables (i.e., the channel coefficients) involved in the state equation are complex Gaussian distributed.
    To tackle this difficulty, we propose a novel approximation method  to approximate the underlying  non-Gaussian variables with Gaussian ones, based on which the state equation is transformed into a linear and Gaussian system.
Then, a generalized KF (GKF)-based channel tracking algorithm is proposed, and by combining it with the LSTM-based channel prediction algorithm, we show that the  2SCTP scheme is also effective in the considered more general case.


\item Finally, extensive numerical simulations are presented to demonstrate the effectiveness of the proposed 2SCTP scheme.
    We show that by exploiting the temporal correlation  of the channel, much lower channel training overhead can be achieved by the proposed 2SCTP scheme as compared to the existing channel estimation methods.
\end{itemize}

The rest of the paper is organized as follows. {  Section II
presents the system model and   the framework of the proposed 2SCTP scheme.}
Then, the proposed CTP algorithms for the special and general cases  are respectively  given in Section III and Section IV.  Numerical results are presented in Section V, and finally Section VI concludes the paper.

Notations:
Scalars, vectors and matrices are respectively denoted by lower (upper) case, boldface lower case and boldface upper case letters. For a matrix $\mathbf{X}$ of arbitrary size, $\mathbf{X}^T$, $\mathbf{X}^*$, $\mathbf{X}^H$ denote the transpose, conjugate and conjugate transpose  of $\mathbf{X}$, respectively. {  $[\mathbf{x}]_{n}$ represents the $n$-th element of the vector $\mathbf{x}$, and $\mathbf{X}_n$ denotes the $n$-th column of the matrix $\mathbf{X}$.} The symbol $||\cdot||$ denotes the Euclidean norm of a complex vector, and $|\cdot|$ is the absolute value of a complex scalar.
$\text{diag}(x_1,\cdots, x_N)$ denotes a diagonal matrix whose diagonal
elements are set as $x_1,\cdots,x_N$. ${\mathbb{C}}^{m\times n}$  denotes the space of $m \times n$ complex matrices. {   The notations $\mathbb{E}\big[\cdot\big]$ and  $\text{Var}\big[\cdot \big]$ represent the expectation and  variance of a random variable,  and  $\text{Cov}\big[\cdot,\cdot\big]$ denote  the covariance of two random variables. The remainder operation is  denoted by \%, i.e., $x_1\% x_2 = x_1 - x_2\lfloor \frac{x_1}{x_2}\rfloor$ with $\lfloor x\rfloor$ denoting the maximum integer that is smaller than $x$.}
The symbol $j$ is used to represent $\sqrt{-1}$.
 Finally, we define the complex Gaussian distribution with mean $\mu$ and variance ${\sigma}^2$  as $\mathcal{C}\mathcal{N}(\mu,{\sigma}^2)$.

\vspace{-0.5cm}\section{System Model and Transmission Protocol}\label{SecII}
\vspace{-0.0cm}\subsection{System Model}
\begin{figure}[b]
\vspace{-0.8cm}
\setlength{\belowcaptionskip}{-0.8cm}
\renewcommand{\captionfont}{\small}
\centering
\includegraphics[scale=.45]{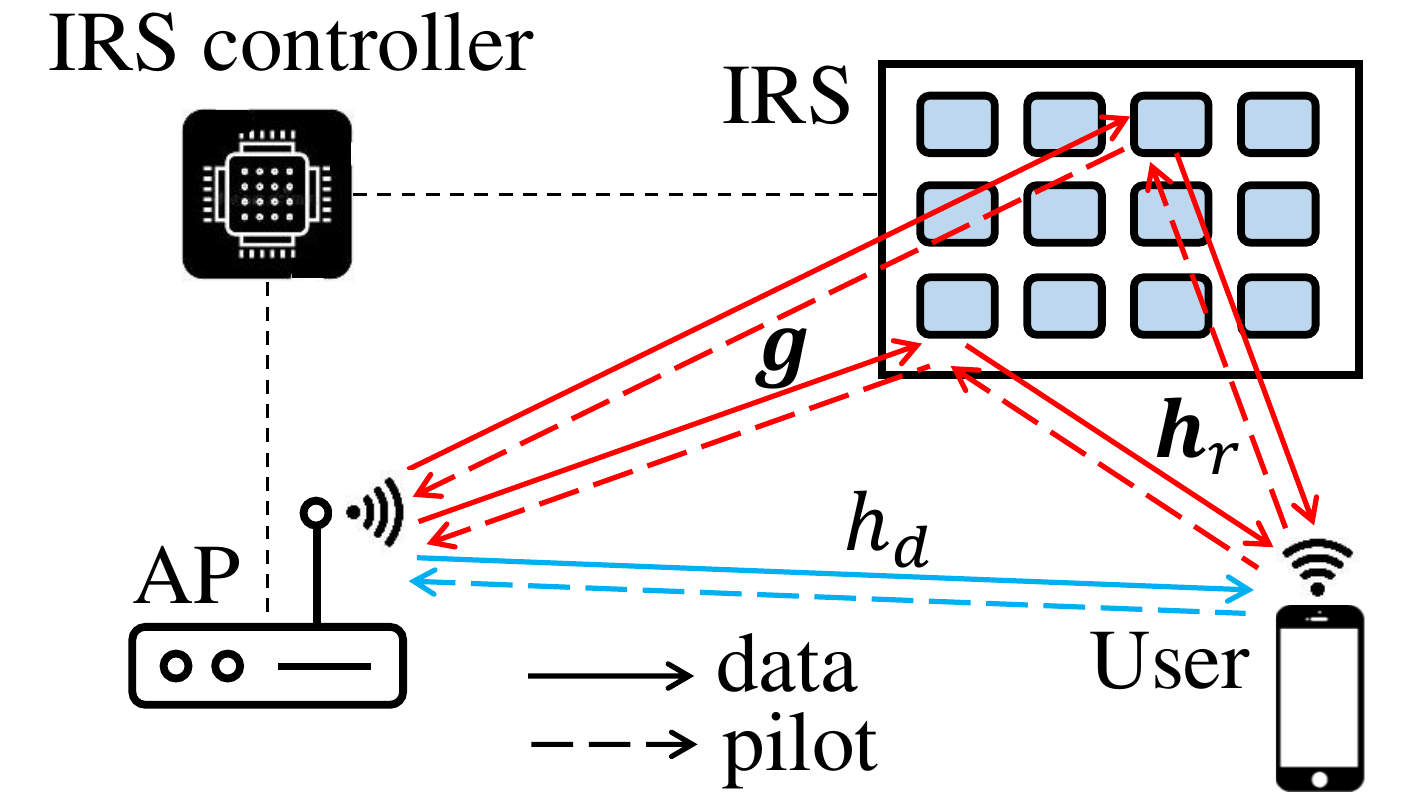}
\caption{  System model of the considered IRS-aided communication system.}
\label{SystemModel}
\normalsize
\end{figure}\begin{figure}[b]
\vspace{-0.8cm}
\setlength{\belowcaptionskip}{-0.5cm}
\renewcommand{\captionfont}{\small}
\centering
\includegraphics[scale=.43]{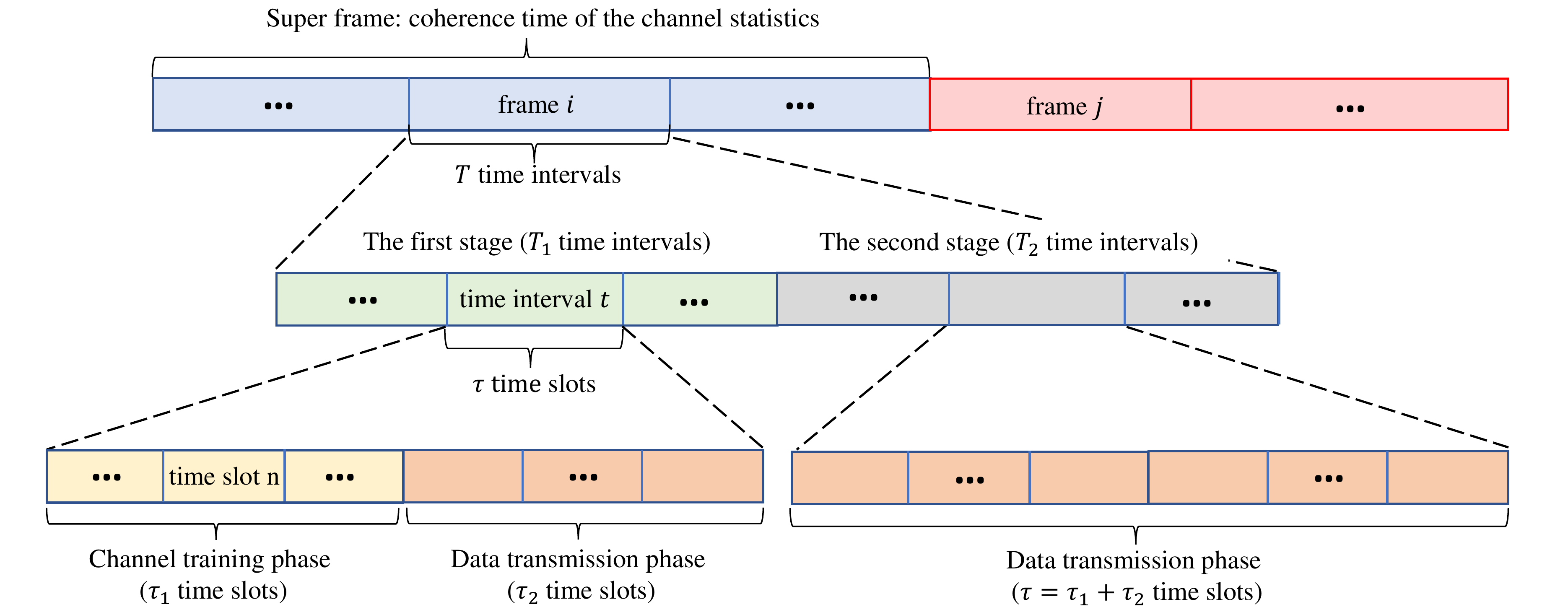}
\caption{  Illustration of the proposed frame structure.}
\label{frame}
\normalsize
\end{figure}
{  As shown Fig. \ref{SystemModel}, we consider an IRS-aided time-division duplexing (TDD) system where a single-antenna AP communicates with a single-antenna user via an IRS with $N$ reflecting elements.
 The CSI is obtained via uplink pilot transmission in the considered system  with the assumption of channel reciprocity.}
 Let $\mathbf{g}\in \mathbb{C}^{N \times 1}$, $\mathbf{h}_r \in \mathbb{C}^{N\times 1}$  and $h_d$ denote the   baseband equivalent IRS-AP, user-IRS and user-AP channels, respectively.
 {  Assume that  the channel statistics  will stay unchanged for a long period of time called a super frame, which includes a  number of frames.}
Each frame  is divided into $T$  time intervals and each time interval further consists of $\tau$ time slots, as illustrated in Fig. \ref{frame}. It is assumed that all the channels remain approximately constant in each time interval and vary from the current
time interval to the next.
Due to the insufficient angular spread of the scattering environment and closely spaced
antennas/reflecting elements, both  line-of-sight (LoS) and non-LoS (NLoS) components may exist in practical channels. As a result, the IRS-AP, user-IRS and user-AP channels in the $t$-th time interval can be respectively modeled as follows:
\vspace{-0.2cm}
\begin{equation}
\mathbf{g}(t) = \sqrt{\frac{l_{\text{IA}}\beta_{\text{IA}}}{1+\beta_{\text{IA}}}}\mathbf{g}^{\text{LoS}} +
\sqrt{\frac{l_{\text{IA}}}{1+\beta_{\text{IA}}}}\mathbf{g}^{\text{NLoS}},
\end{equation}
\vspace{-0.5cm}
\begin{equation}
\mathbf{h}_r(t) = \sqrt{\frac{l_{\text{UI}}\beta_{\text{UI}}}{1+\beta_{\text{UI}}}}\mathbf{h}_r^{\text{LoS}} +
\sqrt{\frac{l_{\text{UI}}}{1+\beta_{\text{UI}}}}\mathbf{h}_r^{\text{NLoS}},
\end{equation}
\vspace{-0.5cm}
\begin{equation}
{h}_d(t) = \sqrt{\frac{l_{\text{UA}}\beta_{\text{UA}}}{1+\beta_{\text{UA}}}}{h}_d^{\text{LoS}} +
\sqrt{\frac{l_{\text{UA}}}{1+\beta_{\text{UA}}}}{h}_d^{\text{NLoS}},
\end{equation}
where $\mathbf{g}^{\text{LoS}}$, $\mathbf{h}_r^{\text{LoS}}$ and ${h}_d^{\text{LoS}}$ represent the LoS components; $\mathbf{g}^{\text{NLoS}}$, $\mathbf{h}_r^{\text{NLoS}}$ and ${h}_d^{\text{NLoS}}$ denote the NLoS components; $\beta_{\text{IA}}$, $\beta_{\text{UI}}$ and $\beta_{\text{UA}}$ are the Rician factors of the IRS-AP, user-IRS and user-AP channels, respectively.
$l_{\text{IA}}$, $l_{\text{UI}}$ and $l_{\text{UA}}$ are the corresponding path losses, which are given by
$l_{\text{IA}} =l_{0}\left({d_{\text{IA}}}/{d_{0}}\right)^{-\gamma^{\text{IA}}}$, $l_{\text{UI}} =l_{0}\left({d_{\text{UI}}}/{d_{0}}\right)^{-\gamma^{\text{UI}}}$ and $l_{\text{UA}} =l_{0}\left({d_{\text{UA}}}/{d_{0}}\right)^{-\gamma^{\text{UA}}}$, respectively, where
$d_0$ represents the  reference distance and $l_0$ is the path loss at the reference distance, ${d_{\text{IA}}}$, ${d_{\text{UI}}}$ and ${d_{\text{UA}}}$ denote the link distances from the IRS to the AP, from the user to the IRS, and from the user to the AP, respectively; $\gamma^{\text{IA}}$, $\gamma^{\text{UI}}$ and $\gamma^{\text{UA}}$ denote the path-loss  exponents.
Equivalently,
we have $\mathbf{g}(t) \sim \mathcal{CN}(\bar{\mathbf{g}}, \mathbf{C}_{\text{IA}})$, $\mathbf{h}_r(t) \sim \mathcal{CN}(\bar{\mathbf{h}}_r, \mathbf{C}_{\text{UI}})$ and
${h}_d(t) \sim \mathcal{CN}(\bar{{h}}_d, {C}_{\text{UA}})$, where $\mathbf{C}_{\text{IA}} = \frac{l_{\text{IA}}}{1+\beta_{\text{IA}}} \mathbb{E}\big[\mathbf{g}^{\text{NLoS}}(\mathbf{g}^{\text{NLoS}})^H\big]$, $\mathbf{C}_{\text{UI}} = \frac{l_{\text{UI}}}{1+\beta_{\text{UI}}} \mathbb{E}\big[\mathbf{h}_r^{\text{NLoS}}(\mathbf{h}_r^{\text{NLoS}})^H\big]$ and
${C}_{\text{UA}} = \frac{l_{\text{UA}}}{1+\beta_{\text{UA}}} \mathbb{E}\big[{h}_d^{\text{NLoS}}({h}_d^{\text{NLoS}})^*\big]$.

Due to the mobility of the AP and/or user, the channels are usually time varying and exhibit a high degree of temporal correlation\cite{Telatar}, \cite{Borade2012}.
Therefore, we   employ
the independent stationary steady-state Gauss-Markov process \cite{Ziniel2013}, \cite{Lian2019} to model the temporal evolution of the channel parameters, which are given as follows:
\vspace{-0.4cm}
\begin{subequations}\label{CM}
\begin{equation}
\mathbf{g}(t) = \sqrt{1-\alpha_{\text{IA}}}\big(\mathbf{g}(t-1) - \bar{\mathbf{g}}\big)+ \sqrt{\alpha_{\text{IA}}}\mathbf{u}_{\textrm{IA}}(t) + \bar{\mathbf{g}},
\end{equation}
\vspace{-1.0cm}
\begin{equation}
\mathbf{h}_r(t) = \sqrt{1-\alpha_{\text{UI}}}\big(\mathbf{h}_r(t-1) - \bar{\mathbf{h}}_r\big) + \sqrt{\alpha_{\text{UI}}}\mathbf{u}_{\textrm{UI}}(t) + \bar{\mathbf{h}}_r,
\end{equation}
\vspace{-1.0cm}
\begin{equation}
{h}_d(t) = \sqrt{1-\alpha_{\text{UA}}}\big({h}_d(t-1)-\bar{h}_d\big) + \sqrt{\alpha_{\text{UA}}}{u}_{\textrm{UA}}(t)+\bar{h}_d,
\end{equation}
\end{subequations}

\vspace{-0.2cm}\noindent
where $\mathbf{u}_{\text{IA}}(t)\sim \mathcal{CN}(\mathbf{0},\sigma^2_{\text{IA}}\mathbf{I})$, $\mathbf{u}_{\text{UI}}(t)\sim \mathcal{CN}(\mathbf{0},\sigma^2_{\text{UI}}\mathbf{I})$ and
${u}_{\text{UA}}(t)\sim \mathcal{CN}({0},{\sigma}^2_{\text{UA}}\mathbf{I})$ are the perturbation terms in the IRS-AP, user-IRS and user-AP channels, respectively; $\alpha_{\text{IA}}$, $\alpha_{\text{UI}}$ and $\alpha_{\text{UA}}$ represent the temporal correlation coefficients of the IRS-AP, user-IRS and user-AP channels, respectively.
Note that the channel realizations   $\mathbf{g}(t-1)$, $\mathbf{h}_r(t-1)$ and ${h}_d(t-1)$ are statistically independent of $\mathbf{u}_{\text{IA}}(t)$, $\mathbf{u}_{\text{UI}}(t)$ and $u_{\text{UA}}(t)$.


In the $t$-th time interval, the received signal at the AP can be expressed as
\vspace{-0.2cm}
\begin{equation}\label{Rece}
y(t) = \sqrt{p}[\mathbf{h}_r^H(t)\bm{\Theta}(t)\mathbf{g}(t) + h_d(t)]s(t) + z(t),
\end{equation}

\vspace{-0.2cm}\noindent
where $s(t)$ denotes the transmit symbol, $p$ represents the transmit power, $z(t)$ denotes the additive white Gaussian
noise (AWGN) with zero-mean and variance $\sigma^2$,  $\bm{\Theta}(t)$ represents the reflection pattern at the IRS.
Let us define  $\mathbf{h}\triangleq [h_d; \text{diag}(\mathbf{h}_r^H)\mathbf{g}]$ as the user-AP  equivalent channel, then the received signal in \eqref{Rece} can be equivalently rewritten as
\vspace{-0.3cm}
\begin{equation}
y(t) = \sqrt{p}\mathbf{v}^H(t)\mathbf{h}(t)s(t) + z(t),
\end{equation}

\vspace{-0.4cm}\noindent
where $\mathbf{v}(t) \triangleq [1,\theta_1(t), \cdots, \theta_N(t)]^H$ with $\theta_n(t)$ denoting the $n$-th diagonal element  of the reflection pattern $\bm{\Theta}(t)$.
Notice that when the number of IRS reflecting elements is large, the acquisition of  the instantaneous  CSI  may cause considerable channel training/estimation overhead, which leads to reduced user transmission rate due to the limited time left for data transmission. To address this issue, we propose to exploit the temporal correlation  of the channel in this work  and present a 2SCTP scheme, where a KF-based channel tracking algorithm and a DL-enabled channel prediction   algorithm are integrated to reduce the channel training overhead and thus improve the transmission rate.

\vspace{-0.4cm}\subsection{Transmission Protocol}
The proposed two-stage transmission protocol is shown in Fig. \ref{frame}.  As can be seen, in the first stage that contains
 $T_1$ time intervals, each time interval is divided into two phases, i.e., the channel training phase (consists of $\tau_1$ time slots) and the data transmission phase (consists of $\tau_2 = \tau - \tau_1$ time slots).
More specifically, in  the first $\tau_1$ time slots of the first stage, the user sends pilot symbols that are known at the AP such that the AP can track the channel according to the received pilot signals, then  the remaining  $\tau-\tau_1$ time slots of each time interval are used for data transmission from the AP to the user.
In the second stage which includes {  $T_2 = T - T_1$} time intervals, there is no channel training phase and the channels are predicted such that all the time slots in these time intervals are utilized for data transmission.

As compared with the existing channel estimation methods, the proposed 2SCTP scheme is able to reduce  the channel training overhead in the following two ways:
\begin{itemize}
\item In the first stage, we  exploit the temporal correlation of the channels to reduce the channel training overhead in each time interval, thus $\tau_1$ can be much smaller than $N+1$.

\item In the second stage,    we predict the channels based on the statistical channel information extracted from the CSI  collected in the first stage, hence the channel training overhead can be completely  removed.
\end{itemize}
Based on the proposed transmission protocol, we propose the 2SCTP scheme, whose details are given in    the following two sections. {  Note that
at the beginning of each super frame, we need to rerun the proposed 2SCTP scheme to track the channel since the channel statistics change. Therefore, in the rest of the paper, we focus on the CTP algorithm design within one frame.
Besides,   the temporal correlation coefficients $\alpha_{\text{IA}}$, $\alpha_{\text{UI}}$ and $\alpha_{\text{UA}}$, the variances of    the perturbation terms $\sigma^2_{\text{IA}}$, $\sigma^2_{\text{UI}}$ and ${\sigma}^2_{\text{UA}}$, and the channel means $\bar{\mathbf{g}}$, $\bar{\mathbf{h}}_r$ and $\bar{h}_d$ are regarded as   prior knowledge of the proposed 2SCTP scheme.\footnote{  These parameters can be obtained by efficient  parameter estimation methods, e.g., the expectation-maximization (EM) algorithm \cite{mclachlan2007algorithm}, and further investigation is left for future work.}}
\vspace{-0.5cm}\section{2SCTP Scheme for the Special Case}\label{Sec3}
In this section, we focus on the  special case where   the
IRS-AP channel is assumed to change much slower than   the user-IRS and user-AP channels, i.e., $\alpha_{\text{IA}} \ll \alpha_{\text{UI}}$ and  $\alpha_{\text{IA}} \ll \alpha_{\text{UA}}$.
Furthermore, for simplicity, we assume  that the  user-AP, IRS-AP and user-IRS channels follow the Rayleigh channel model, extension to the more general Rician channel model will be discussed later. Based on these assumptions, the user-AP  equivalent    channel $\mathbf{h}(t)$ can be modeled as
\vspace{-0.3cm}
\begin{equation}\label{S1}
\mathbf{h}(t) = \mathbf{A}_h\mathbf{h}(t-1)+\mathbf{B}_h\mathbf{u}(t),
\end{equation}

\vspace{-0.2cm}\noindent
where
\begin{equation}
\mathbf{A}_h = \left[\begin{array}{cc}
\sqrt{1-\alpha_{\text{UA}}}  &  \mathbf{0}_{1\times N}\\
\mathbf{0}_{N\times 1} & \sqrt{1-\alpha_{\text{UI}}}\mathbf{I}
\end{array}
\right],
\mathbf{B}_h = \left[\begin{array}{cc}
\sqrt{\alpha_{\text{UA}}}  &  \mathbf{0}_{1\times N}\\
\mathbf{0}_{N\times 1} & \sqrt{\alpha_{\text{UI}}} \mathbf{I}
\end{array}
\right],
\end{equation}

\vspace{-0.2cm}\noindent
$\mathbf{u}(t)\sim \mathcal{CN}(\mathbf{0},\mathbf{C}_h)$ and $\mathbf{C}_h$ is given by
\vspace{-0.3cm}
\begin{equation}
\mathbf{C}_h = \left[
\begin{array}{cc}
\sigma^2_{\text{UA}} & \mathbf{0}_{1 \times N} \\
\mathbf{0}_{N\times 1} & \sigma^2_{\text{UI}}\text{diag}(\mathbf{g} \mathbf{g}^H )
\end{array} \right].
\end{equation}

\vspace{-0.2cm}\noindent
In this special case, although  the reflected channel $\mathbf{g}^H(t)\text{diag}\big(\mathbf{h}_r(t)\big)$ is related with the IRS-AP and user-IRS channels, it can be assumed to change according to one temporal correlation coefficient $\alpha_{\text{UI}}$ since $\alpha_{\text{IA}} \ll \alpha_{\text{UI}}$. At one extreme, if $\alpha_{\text{UA}} = \alpha_{\text{UI}} = 0$, the reflected channel parameters are totally correlated, (i.e., $\mathbf{h}(t) = \mathbf{h}(t-1)$),
while at the other extreme when $\alpha_{\text{UA}} = \alpha_{\text{UI}} = 1$, the reflected channel coefficients  evolve according to an uncorrelated Gaussian random process.

\vspace{-0.4cm}\subsection{KF-based Channel Tracking in Both Stages}
As a well-known and powerful  variable estimation method,  the KF method has been widely applied in time series analysis, such as signal processing and econometrics \cite{Paul2000,Ghysels2018}. The KF method  keeps track of the estimated state of the system and the uncertainty of the estimate, and only the estimated state from the previous time step and the current measurement are required to obtain  the current state estimate and the corresponding  uncertainty.
Inspired by the superiority of the KF method  in time series processing, we propose a KF-based channel tracking algorithm to estimate the time-varying channels    for the considered special case with limited channel training overhead. {Notice that the proposed channel tracking algorithm is employed in both stages.}

 {First of all, let us introduce the problem setting for a general KF method. The KF model assumes that the state of a system at time step $t$ evolved
from the prior state at time step $t-1$ according to the linear state equation
$
{x}_{t}={F}_{t} {x}_{t-1}+{w}_{t},
$
where $F_t$ is the state transition model, the state variable $x_t$ and state noise variable $w_t$ follow independent complex Gaussian distributions. Then, the observation at time step $t$ is obtained according to the linear observation equation
$
{z}_{t}={H}_{t} {x}_{t-1}+{r}_{t},
$
where $H_t$ is the measurement matrix, the observation variable $z_t$ and observation noise variable $r_t$ also follow complex Gaussian distributions. The KF method aims to  predict the  state variable by utilizing a series of observations.
}

{To apply the KF method to the considered  channel tracking problem, we can regard \eqref{S1}  as the  state equation.  {By  assuming that the pilot symbol is set   as  $s=1$ without loss of generality, the observation functions in the first and second stages are respectively given by
\vspace{-0.3cm}
\begin{equation}\label{y}
\mathbf{y}(t) = \sqrt{p}\mathbf{V}(t)\mathbf{h}(t) + \mathbf{z}(t),  \;\;\; t = 1,\cdots,T_1,
\end{equation}
\vspace{-1.0cm}
\begin{equation}\label{yi}
{\mathbf{y}}_i(t) = \sqrt{p}\mathbf{V}(t)\mathbf{h}(t) + \mathbf{z}(t),  \;\;\; t = T_1+1,\cdots,T,
\end{equation}

\vspace{-0.2cm}\noindent
 where the received signal $\mathbf{y}(t)\in \mathbb{C}^{\tau_1\times 1}$ and ${\mathbf{y}}_i(t)\in \mathbb{C}^{\tau_1\times 1}$ serve as the real and imaginary observations, respectively}, $\mathbf{V}(t) \triangleq [\mathbf{v}_1(t),\cdots, \mathbf{v}_{\tau_1}(t)]^T \in \mathbb{C}^{\tau_1\times (N+1)}$ is regarded as the  measurement matrix, and $\mathbf{v}_{i}(t)$ represents the reflection pattern in the $i$-th time slot of the $t$-th time interval.
{Note that the \emph{imaginary observation} is defined as the received signal if we imagine that the user sends the pilot symbol $1$ with the known reflection patterns  for $t = T_1 +1, ..., T$, and the details  of how to obtain (predict) the imaginary observations in $t = T_1 + 1,..., T$ from the real observations in $t = 1,...,T_1$ using the LSTM-based neural network will be elaborated    in Section \ref{OB}.}
Furthermore,   the value of  $\tau_1$, i.e.,  the dimension of the received signal $\mathbf{y}(t)$, will affect   the convergence of the proposed KF-based channel tracking algorithm and the required channel training overhead, thus it should be carefully selected.
{
Besides, to facilitate the prediction of the imaginary observations using the LSTM-based neural network and improve   the channel tracking performance, we employ a
periodic measurement matrix which satisfies
\vspace{-0.3cm}
\begin{equation}\label{Q}
\mathbf{V}(t) = [\mathbf{Q}_{(\tau_1(t-1)+1)\%(N+1)}, \cdots, \mathbf{Q}_{(\tau_1t)\%(N+1)}]^T,
\end{equation}

\vspace{-0.2cm}\noindent
where  each column is chosen  from a full-rank reference matrix $\mathbf{Q}\in \mathbb{C}^{(N+1)\times (N+1)}$. The design of the measurement matrix will be discussed in Section \ref{OB}, and
the effect of choosing different   $\tau_1$ and employing different reference matrices  will be investigated in Section \ref{Sec5}.}
{For clarity, we introduce
\vspace{-0.1cm}
\begin{equation}
\tilde{\mathbf{y}}(t) = \left\{ \begin{array}{ll}
 \mathbf{y}(t), &t = 1,\cdots, T_1, \\
 \mathbf{y}_i(t), &t = T_1+1,\cdots, T.
\end{array}\right.
\end{equation}
Let us define    $\mathbf{y}_t  \triangleq [\tilde{\mathbf{y}}^T(t),\tilde{\mathbf{y}}^T(t-1),\cdots,\tilde{\mathbf{y}}^T(1)]^T$ as the vector of observations and define $
\hat{\mathbf{y}}(t) \triangleq \tilde{\mathbf{y}}(t) - \mathbf{V}(t)\hat{\mathbf{h}}(t)$
as the innovation
with $
\hat{\mathbf{h}}(t) = \mathbb{E}\big[\mathbf{h}(t)|\mathbf{y}_t\big]$ denoting the predicted channel vector.} Then, the minimum mean square error (MMSE) estimate of $\mathbf{h}(t)$ can be recursively obtained by the KF equations \cite{Fundamentals1993}, where each iteration works in the following two-step process.
\begin{itemize}
\item In the {prediction step}, the estimate of the current state variable $\hat{\mathbf{h}}(t)$ along with its covariance matrix $\mathbf{M}(t)$ are predicted using the state estimated from the previous time interval, i.e.,
    \vspace{-0.3cm}
    \begin{equation}\label{KF1}
    \hat{\mathbf{h}}(t) = \mathbf{A}_h\mathbf{h}_{\text{KF}}(t-1),
    \end{equation}
    \vspace{-1cm}
    \begin{equation}\label{KF2}
    \mathbf{M}(t) = \mathbb{E}\big[\big(\mathbf{h}(t)-\hat{\mathbf{h}}(t) \big)\big( \mathbf{h}(t)-\hat{\mathbf{h}}(t)\big)^H  \big]
= \mathbf{A}_h\mathbf{M}_{\text{KF}}(t-1)\mathbf{A}_h^T + \mathbf{B}_h\mathbf{C}_h\mathbf{B}_h^T,
    \end{equation}

    \vspace{-0.2cm}\noindent
    where $\mathbf{h}_{\text{KF}}(t)$ represents the  refined state estimate in the $t$-th time interval (also named as
\emph{correction}), which will be introduced later,  and  $\mathbf{M}_{\text{KF}}(t) = \mathbb{E}\big[ \big(\mathbf{h}(t) - \mathbf{h}_{\text{KF}}(t-1)   \big)\big( \mathbf{h}(t) - \mathbf{h}_{\text{KF}} (t-1) \big)^H  \big]$ denotes the estimation covariance matrix.
\item In the {update step}, we first update the Kalman gain $\mathbf{G}_{\text{KF}}(t)$ as
    \vspace{-0.3cm}
\begin{equation}\label{KF3}
\mathbf{G}_{\text{KF}}(t) = \mathbf{M}(t)\mathbf{V}^H(t) \big(
\mathbf{V}(t)\mathbf{M}(t)\mathbf{V}^H(t)+\sigma^2\mathbf{I}\big)^{-1}.
\end{equation}

\vspace{-0.2cm}\noindent
Then, multiplying  the innovation $\hat{\mathbf{y}}(t)$   by the Kalman gain $\mathbf{G}_{\text{KF}}(t)$ and combining $\hat{\mathbf{y}}(t)$ with the state estimate $\hat{\mathbf{h}}(t)$, the correction $ \mathbf{h}_{\text{KF}}(t)$ can be obtained via
\vspace{-0.3cm}
\begin{equation}\label{KF4}
\mathbf{h}_{\text{KF}}(t)= \mathbb{E}\big[\mathbf{h}(t)|\mathbf{y}_{t-1}\big]
+ \mathbb{E}\big[\mathbf{h}(t)|\hat{\mathbf{y}}(t)\big] = \hat{\mathbf{h}}(t) + \mathbf{G}_{\text{KF}}(t)\hat{\mathbf{y}}(t).
\end{equation}

\vspace{-0.2cm}\noindent
 Finally,  the estimation covariance matrix $\mathbf{M}_{\text{KF}}(t)$    required to calculate the prediction covariance matrix $\mathbf{M}(t)$ in \eqref{KF3}, is updated by
 \vspace{-0.3cm}
 \begin{equation}\label{KF5}
 \mathbf{M}_{\text{KF}}(t) = \mathbb{E}\big[ \big(\mathbf{h}(t) - \mathbf{h}_{\text{KF}}(t-1)   \big)\big( \mathbf{h}(t) - \mathbf{h}_{\text{KF}} (t-1) \big)^H  \big] =\mathbf{M}(t) - \mathbf{G}_{\text{KF}}(t) \mathbf{V}^H(t)\mathbf{M}(t).
 \end{equation}
\end{itemize}

\vspace{-0.2cm}
To summarize,  the proposed  KF-based channel tracking algorithm is shown in Algorithm \ref{Kalman}, where we set $\mathbf{h}_{\text{KF}}(-1) = \mathbf{0}$ and $ \mathbf{M}_{\text{MF}} = \mathbb{E}\big[\mathbf{h}\mathbf{h}^H\big]$ to initialize the KF iterations.
Besides, Algorithm \ref{Kalman}  can be readily extended to the Rician channel model. Specifically, let $\bar{\mathbf{h}}$ represent the mean of $\mathbf{h}(t)$, and by substituting $\tilde{\mathbf{h}}(t) = \mathbf{h}(t) - \bar{\mathbf{h}}$ into \eqref{S1}, we obtain
\vspace{-0.3cm}
\begin{equation}\label{KalmanRicain}
\tilde{\mathbf{h}}(t) = \mathbf{A}_h\tilde{\mathbf{h}}(t-1)+\mathbf{B}_h\mathbf{u}(t).
\end{equation}

\vspace{-0.3cm}\noindent
 Regarding \eqref{KalmanRicain} as the state equation, the estimation of the state variable in the $t$-th time interval, i.e.,  $\tilde{\mathbf{h}}_{\text{KF}}(t)$,  can be obtained by directly  applying Algorithm \ref{Kalman}, after which we can acquire the refined state estimate in the  $t$-th time interval as  $\mathbf{h}_{\text{KF}}(t) = \tilde{\mathbf{h}}_{\text{KF}}(t) + \bar{\mathbf{h}}$.

\begin{algorithm}[t]
\setlength{\belowcaptionskip}{-2cm}
\setlength{\abovecaptionskip}{-2cm}
\caption{KF-based Channel Tracking Algorithm} 
\label{Kalman}
\hspace*{0.02in} {\bf Input:} covariance matrix  $\mathbf{C}_h = \mathbb{E}\big[\mathbf{h}\mathbf{h}^H\big]$ and observations {  $\tilde{\mathbf{y}}(t)$} \\ 
\hspace*{0.02in} {\bf Output:} {  Estimated channel vectors $\mathbf{h}_{\text{KF}}(t)$}
\begin{algorithmic}[1]
\State \textbf{Initialization:} $t = 0$, $\mathbf{h}_{\text{KF}}(-1) = \mathbf{0}$, $ \mathbf{M}_{\text{MF}} = \mathbb{E}\big[\mathbf{h}\mathbf{h}^H\big]$
\While {{  $t >= 1$}}
\State \emph{Prediction step}:

Calculate the prediction $\hat{\mathbf{h}}(t)$  according to  \eqref{KF1},

Calculate the prediction covariance matrix $\mathbf{M}(t)$ according to  \eqref{KF2},

\State \emph{Update step}:

Update the Kalman gain vector $\mathbf{G}_{\text{KF}}(t)$ according to  \eqref{KF3},

Calculate the correction $\mathbf{h}_{\text{KF}}(t)$ according to  \eqref{KF4},

Calculate the estimation covariance matrix $\mathbf{M}_{\text{KF}}(t)$ according to  \eqref{KF5},

\State $t = t+1$,
\EndWhile
\end{algorithmic}
\end{algorithm}

\vspace{-0.4cm}\subsection{DL-enabled Channel Prediction in the Second Stage}\label{OB}
Solving the channel prediction problem can be interpreted as a time series processing task.
Recently,  DL techniques  designed for time series processing  have been widely used to solve channel prediction problems, e.g.,
using recurrent neural network (RNN) for frequency-domain channel prediction \cite{8746352},
deploying  LSTM networks to predict  fading channels \cite{Prediction9044427}, etc.
However, predicting a high-dimensional channel vector   in massive MIMO systems or IRS-aided communication systems is not easy, and it is challenging to achieve high prediction accuracy. Moreover, in order to process the high-dimensional data, the prediction neural network requires a huge number of neurons/layers to ensure adequate network capacity. This will lead to excessive computing resources and memory costs.
{Assume that the varying reflection patterns $\mathbf{v}(t),t = 1,\cdots,T$, are known at the  AP and user, we can obtain the {imaginary observation} $\mathbf{y}_i(t)$ in the $t$-th time interval  according to \eqref{yi}.}
Then, to tackle  the abovementioned
challenges, we propose to design a novel OB-LSTM network to predict the low-dimensional imaginary observations, instead of the high-dimensional channel vector.
 These predicted imaginary observations are then regarded as the input of the KF-based channel tracking algorithm, which in turn outputs the current channel vector.
 Based on the above strategy,  no pilot symbols are required to be transmitted in the second stage, thus the channel training overhead is zero in this stage and the transmission rate of the user is maximized.


The detailed structure of the proposed OB-LSTM network is provided as follows. {First, let $L_I$ and $L_P$ denote
the lengths of the input observations and predicted imaginary observations, respectively, where $L_I\le T_1$ and $L_P\le T_2$.}
Then, in the $t$-th time interval for $t = T_1,\cdots,T$, we feed  $L_I$ observations $\{\tilde{\mathbf{y}}(t-L_{I}+1),\cdots, \tilde{\mathbf{y}}(t)\}$ into the OB-LSTM network and train it  to predict the next $L_P$ imaginary observations, i.e., $\{ {\mathbf{y}}_i(t + 1),\cdots, {\mathbf{y}}_i(t+L_P) \}$. Since the observations are related to both the unknown channels and the known reflection matrix $\mathbf{V}(t)$, we also treat $\mathbf{V}(t)$ as the part of the network input.
Besides,  the historical observations are pre-processed by normalization to
improve the performance of the multilayer perceptron models in the proposed OB-LSTM network. After data normalization,  the overall input data sequence is given by  $\{\mathbf{x}(1),\cdots,\mathbf{x}(t),\cdots, \mathbf{x}(L_I)\}$, where
$\mathbf{x}(t) = [F_n(\tilde{\mathbf{y}}^T(t)),F_n([\mathbf{v}_1^T(t),\cdots,$ $\mathbf{v}_{\tau_1}^T(t) ])]^T \in \mathbb{C}^{D_I\times 1}$ with $D_I \triangleq \tau_1(N+1) $ and $F_n(\cdot)$ denoting the dimension of the input   and the normalization function, respectively. {Accordingly,  we collect a set of real observations and construct the training data-label sample set    as $\big\{ \{\mathbf{x}(j+1),\cdots, \mathbf{x}(j+L_I)\},  \{\mathbf{y}(j+L_I + 1),\cdots, \mathbf{y}(j+L_I+L_P)\} \big\}_{j=0}^{J-1}$.}


Different from classical   deep learning applications, e.g., computer vision and natural language processing, complex data format is usually  considered in wireless communications. To handle the complex input data, the proposed OB-LSTM network is designed to contain two sub-networks with the same structure, which are called the  R sub-network and the I sub-network, respectively.\footnote{Alternatively, we can combine  the real and imaginary parts of the input  observations   into one vector as the network input, such that no sub-networks are needed.   However, the proposed sub-network structure can achieve similar performance with less network parameters, thus is more favorable here.}
In the following, we only focus on the structure of the R sub-network since the I sub-network can be similarly designed. Specifically,
 each sub-network contains one input fully-connected layer, $K$ LSTM layers and one output fully-connected layer.  The input fully-connected layer is designed to augment   the input dimension of the LSTM units, which can help to extract the high-dimensional features of the input data and improve the prediction performance.
Let $\epsilon$ represent  an expansion  factor, then the output of the input fully-connected layer, denoted by $\mathbf{x}_h(t) \in \mathbb{R}^{\epsilon D_I\times 1}$,   can be expressed as
\vspace{-0.3cm}
\begin{equation}
\mathbf{x}_{h}(t) = \text{ReLU}\big(\mathbf{W}_e\mathbf{x}(t) + \mathbf{b}_e\big),
\end{equation}

\vspace{-0.3cm}\noindent
where $\mathbf{W}_e\in {\mathbb{C}}^{\epsilon D_I\times D_I}$ and $\mathbf{b}_e \in \mathbb{C}^{\epsilon D_I\times 1}$ denote the weight and bias (i.e., the learnable network parameters), respectively, and   $\text{ReLU}(\cdot)$ is employed  as the activation function.
The LSTM layers are the core of the proposed network, since they are developed to deal with the vanishing gradient problem that might be encountered when training traditional RNNs and are applicable to tasks such as classifying, processing and  predicting based on time series data \cite{lstm}.
 In this work, one LSTM layer consists of $L_I$ cascading  LSTM units. Each LSTM unit is composed of a cell, an input gate, an output gate and a forget gate, where the cell remembers information over arbitrary time interval and the three gates regulate the flow of information into and out of the cell. As shown in Fig. \ref{LSTM}, in the $(t,k)$-LSTM unit (i.e., the $t$-th LSTM unit in the $k$-th LSTM layer),  the forget gate $\mathbf{f}_{t,k}$, the input gate $\mathbf{i}_{t,k}$, the output gate $\mathbf{o}_{t,k}$, the cell state $\mathbf{c}_{t,k}$ and the output $\mathbf{y}_t$
  are calculated as follows:
\begin{figure}[t]
\vspace{-0.5cm}
\setlength{\belowcaptionskip}{-0.9cm}
\renewcommand{\captionfont}{\small}
\centering
\includegraphics[scale=.30]{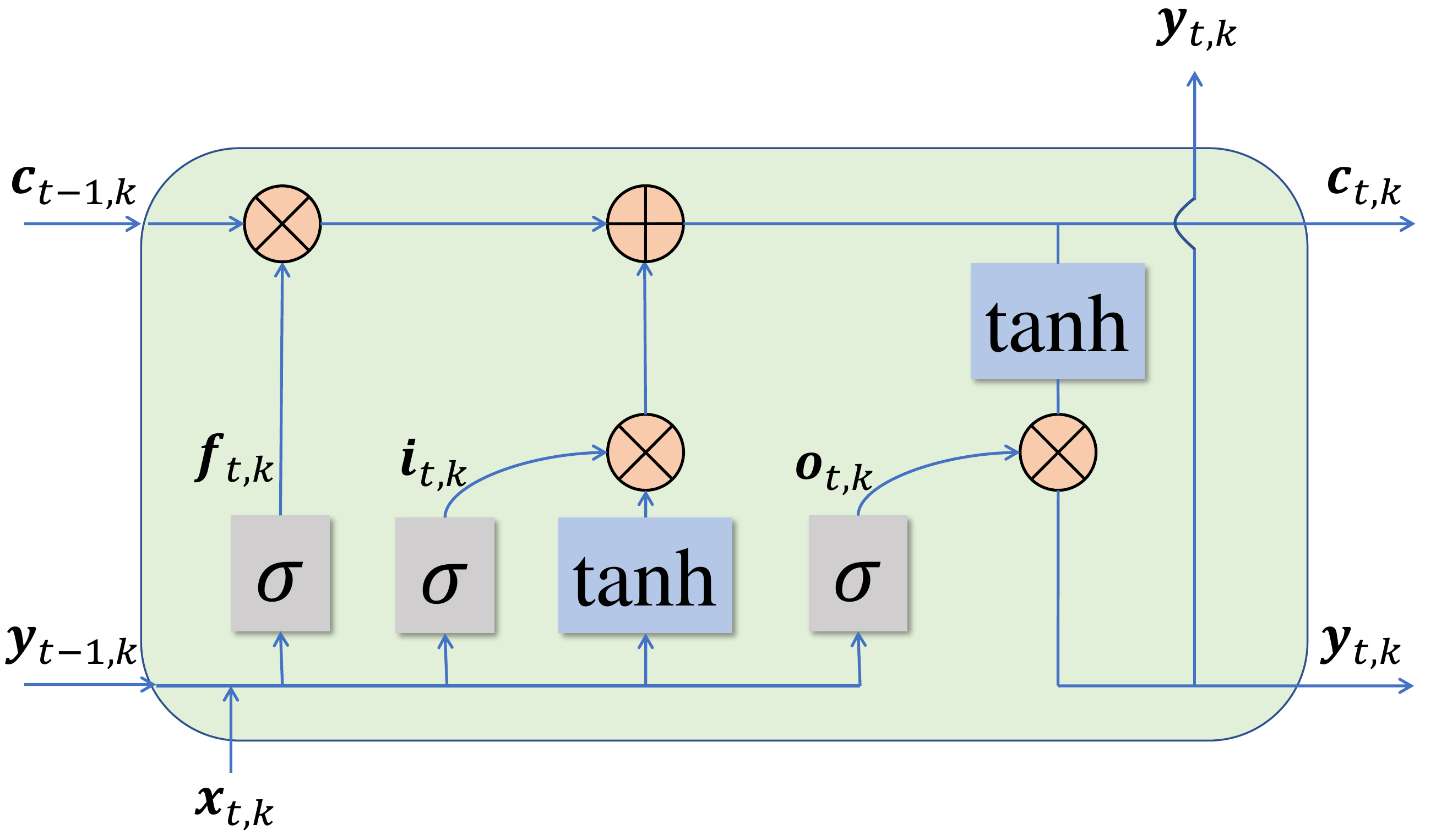}
\caption{The structure of one LSTM unit.}
\label{LSTM}
\normalsize
\end{figure}
\vspace{-0.3cm}
  \begin{equation}
  \begin{aligned}
  &\mathbf{f}_{t,k}=\sigma \left(\mathbf{W}_{t,k}^{fx} \mathbf{x}_{t,k} + \mathbf{W}_{t,k}^{fy} \mathbf{y}_{t-1,k}+\mathbf{b}_{t,k}^{f}\right),\\
  &\mathbf{i}_{t,k}=\sigma\left(\mathbf{W}_{t,k}^{ix} \mathbf{x}_{t,k} + \mathbf{W}_{t,k}^{iy}\mathbf{y}_{t-1,k}+\mathbf{b}_{t,k}^{i}\right),\\
 & \mathbf{o}_{t,k}=\sigma\left(\mathbf{W}_{t,k}^{ox} \mathbf{x}_{t,k} + \mathbf{W}_{t,k}^{oy} \mathbf{y}_{t-1,k}+\mathbf{b}_{t,k}^{o}\right),\\
  &\mathbf{c}_{t,k}=\mathbf{f}_{t,k} \otimes \mathbf{c}_{t-1,k}+\mathbf{i}_{t,k} \otimes \tanh \left(\mathbf{W}_{t,k}^{cx}\mathbf{x}_{t,k} + \mathbf{W}_{t,k}^{cy} \mathbf{y}_{t-1,k}+\mathbf{b}_{t,k}^{c}\right),\\
  &\mathbf{y}_{t,k}=\mathbf{o}_{t,k} \otimes \tanh \left(\mathbf{c}_{t,k}\right),
  \end{aligned}
  \end{equation}

  \vspace{-0.2cm}\noindent
where $\mathbf{x}_{t,k}$ denotes the input of this LSTM unit; $\mathbf{W}_{t,k}^{\kappa} \in  \mathbb{R}^{\epsilon D_I \times 1}, \kappa \in \{fx, fy, ix, iy, ox, oy, cx, cy\}$
and $\mathbf{b}_{t,k}^\nu \in \mathbb{R}^{\epsilon D_I \times 1},
\nu \in \{f,i,o,c\}$
 represent the corresponding  weights and biases, respectively; the sigmoid function  $\sigma(\cdot)$  and the tanh function $\text{tanh}(\cdot)$ are considered as the activation functions of the gates and the cell, respectively.  Note that the  output of $\sigma(\cdot)$ varies in $[0,1]$, hence this activation function describes how much information can pass through this gate.
 Finally, the output fully-connected layer transforms the $L_I$ outputs of the last LSTM layer, i.e., $\mathbf{y}_{1,K},\cdots, \mathbf{y}_{L_I,K}$, into the prediction of the  $L_P$ imaginary observations, i.e.,
 \vspace{-0.3cm}
 \begin{equation}
 \mathbf{y}_p = \mathbf{W}_p[\mathbf{y}_{1,K}^T,\cdots,\mathbf{y}_{L_I,K}^T]^T + \mathbf{b}_p,
 \end{equation}

 \vspace{-0.2cm}\noindent
 where $\mathbf{y}_p \in \mathbb{R}^{\tau_1 L_P\times 1}$ denotes the output, $\mathbf{W}_p\in \mathbb{R}^{\tau_1 L_P \times \epsilon D_I L_I}$ and $\mathbf{b}_p \in  \mathbb{R}^{\tau_1 L_P \times 1}$ represent the corresponding weight and bias. The  outputs of the R and I sub-networks  are regarded as the real and imaginary parts, which are then  utilized to  construct the final $L_P$ predicted imaginary observations $\hat{\mathbf{y}}(L_I+1),\cdots, \hat{\mathbf{y}}(L_I+L_P)$.

 To summarize, the output of the proposed OB-LSTM network can be expressed as
  \vspace{-0.3cm}
 \begin{equation}
 \{\hat{\mathbf{y}}(j+L_I+1), \cdots, \hat{\mathbf{y}}(j+L_I+L_P)\} = \text{OB-LSTM}
 \big(\{\mathbf{x}(j+1),\cdots, \mathbf{x}(j+L_I)\}; \{ \bm{\Theta}_R, \bm{\Theta}_I \}\big),
 \end{equation}

  \vspace{-0.2cm}\noindent
 where $\bm{\Theta}_R$  $(\bm{\Theta}_I)$ represents the collection of network parameters included in the R (I) sub-network.
\begin{figure}[t]
\vspace{-1.0cm}
\setlength{\belowcaptionskip}{-0.8cm}
\renewcommand{\captionfont}{\small}
\centering
\includegraphics[scale=.25]{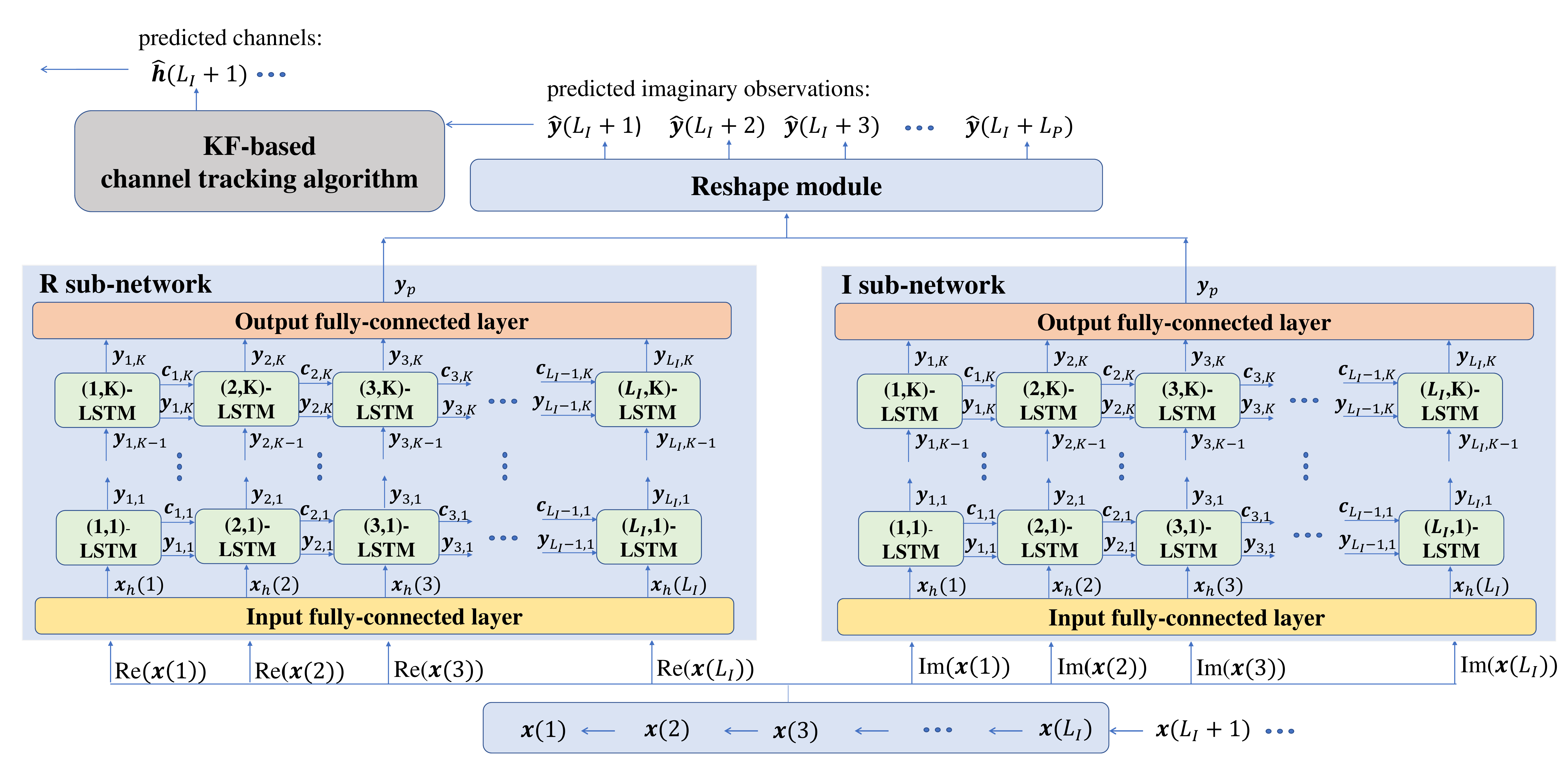}
\caption{Overall procedure of the proposed channel prediction method.}
\label{CPstage}
\normalsize
\end{figure}
The overall procedure of the proposed channel prediction method is depicted in Fig. \ref{CPstage}, where the structure of the proposed OB-LSTM network is also illustrated.

{  In this work, we adopt the supervised learning strategy and construct the following loss function which is  based on the $\ell_2$-norm of the difference between the future observations and their estimates}
 \vspace{-0.3cm}
\begin{equation}\label{L2}
   \mathcal{L}_2(\bm{\Theta}_{R},\bm{\Theta}_{I}) = \frac{1}{J}\sum_{j=0}^{J-1} \sum_{k = 1}^{T_P}\|\hat{\mathbf{y}}(j +L_I+k) -  \mathbf{y}(j+L_I +k) \|_2^2.
\end{equation}

 \vspace{-0.3cm}\noindent
All the network parameters in the proposed OB-LSTM network are updated by minimizing this loss function using the Adam optimizer \cite{Adam}.
{It is noteworthy that the parameters of the  OB-LSTM network are learned through  offline training. When the channel statistics change,  the user should send pilot symbols in the first several frames of a super frame,   as such a set of real observations can be obtained, based on which  we can further finetune the network parameters to alleviate the performance loss caused by model mismatch.}

{  When the training is finished, the OB-LSTM network can be flexibly integrated into the proposed 2SCPT scheme. {We take the case of $T_1 = 6$ and $T_2 = 6$ as a toy example to elaborate this.    The most simple strategy (also called \emph{Strategy A}) is to let $L_I = T_1=6$ and $L_P = T_2=6$, then $6$ real observations, i.e., $\mathbf{y}(1),\cdots, \mathbf{y}(6)$, collected in the first stage can be directly input into  the OB-LSTM network to predict  $6$ imaginary observations, i.e., $\mathbf{y}_i(7),\cdots, \mathbf{y}_i(12)$. However, in this strategy, if the value of  $T_1$ or $T_2$ is large, the
 computational complexity of the OB-LSTM network
 will be high since the input/output dimensions and the number of neurons/layers are large.
 Thus, it would be
 difficult to train such a network with good prediction performance. An alternative strategy (also called \emph{Strategy B}) is to set  $L_P < L_I\le T_1$, e.g., $L_I = 6$ and $L_P = 1$, such that  the predicted imaginary observations can be reused as the network input. Specifically, in the first time interval of the second stage, the OB-LSTM network predicts the next imaginary observation $\mathbf{y}_i(7)$ based on the last 6 real observations $\mathbf{y}(1),\cdots, \mathbf{y}(6)$ collected in the first stage, while in the subsequent time intervals, we can construct the input data using 1  imaginary observation and the last $5$ real observations, i.e., $\mathbf{y}(2),\cdots, \mathbf{y}(6),\mathbf{y}_i(7)$, to obtain $\mathbf{y}_i(8)$.}
Note that the hyper-parameters of the OB-LSTM network, e.g., $L_I$, $L_P$ and $\epsilon$, and the employed  strategies should be carefully chosen  to balance the prediction performance, the computational complexity of the network and the required channel training overhead.  Investigation into the impacts of different   hyper-parameters and strategies will be shown in Section \ref{Sec5}.

\vspace{-0.3cm}
 \begin{remark}
When we consider the case where the AP is deployed with  multiple antennas, the user-AP equivalent channel matrices, i.e., the reflected channel matrices and the direct channel vectors, are required to be tracked and predicted based on the observations.
By vectorizing the channel matrices and constructing the observation functions  as in  \eqref{y} and \eqref{yi}, the proposed 2SCTP scheme can be readily extended to this more general case.
 \end{remark}

\vspace{-0.5cm}
\section{2SCTP Scheme For the General Case}
{  In this section, we extend the proposed 2SCTP scheme to the  more general case where the IRS-AP, user-IRS and user-AP channels are all assumed to be  time-varying. In this case, the user-AP equivalent   channel $\mathbf{h}(t)$ can be written as
\vspace{-0.3cm}
\begin{equation}\label{S2}
\mathbf{h}(t) = [
\sqrt{1-\alpha_{\text{UA}}}{h}_d(t-1) + \sqrt{\alpha_{\text{UA}}}{u}_{\textrm{UA}}(t),
\tilde{{h}}_1(t),
\cdots,
\tilde{{h}}_N(t)]^T,
\end{equation}

\vspace{-0.2cm}\noindent
where $\tilde{h}_n(t), n \in [1,N]$ consists of four random variables and is given by
\vspace{-0.3cm}
\begin{equation}
\begin{aligned}
\tilde{{h}}_n(t) = &\sqrt{(1-\alpha_{\text{IA}})(1-\alpha_{\text{UI}})} [\mathbf{g}(t-1)]_n [\mathbf{h}_r(t-1)]_n +
\sqrt{\alpha_{\text{IA}}(1-\alpha_{\text{UI}})} [\mathbf{u}_{\textrm{IA}}(t)]_n [\mathbf{h}_r(t-1)]_n \\
&+
\sqrt{\alpha_{\text{UI}}(1-\alpha_{\text{IA}})}[\mathbf{g}(t-1)]_n[\mathbf{u}_{\textrm{UI}}(t)]_n +
\sqrt{\alpha_{\text{UI}}\alpha_{\text{IA}}} [\mathbf{u}_{\textrm{IA}}(t)]_n [\mathbf{u}_{\textrm{UI}}(t)]_n.
\end{aligned}
\end{equation}

\vspace{-0.2cm}\noindent
Then,  we can simplify  the state equation \eqref{S2} as
\vspace{-0.3cm}
\begin{equation}\label{S3}
\mathbf{h}(t)  = \mathbf{A}_g\mathbf{h}(t-1) + \mathbf{u}_g(t),
\end{equation}

\vspace{-0.2cm}\noindent
where
\vspace{-0.1cm}
\begin{equation}
\mathbf{A}_g = \left[\begin{array}{cc}
\sqrt{1-\alpha_{\text{UA}}}  &  \mathbf{0}_{1\times N}\\
\mathbf{0}_{N\times 1} & \sqrt{(1-\alpha_{\text{IA}})(1-\alpha_{\text{UI}})}\mathbf{I}
\end{array}
\right],
\end{equation}

\vspace{-0.2cm}\noindent
 $\mathbf{u}_g(t)$ represents the composite noise variable that satisfies $[\mathbf{u}_g(t)]_1 = {u}_{\textrm{UA}}(t)$ and $[\mathbf{u}_g(t)]_n = \sqrt{\alpha_{\text{IA}}(1-\alpha_{\text{UI}})}$ $ [\mathbf{u}_{\textrm{IA}}(t)]_n [\mathbf{h}_r(t-1)]_n +
\sqrt{\alpha_{\text{UI}}(1-\alpha_{\text{IA}})}[\mathbf{g}(t-1)]_n[\mathbf{u}_{\textrm{UI}}(t)]_n +
\sqrt{\alpha_{\text{UI}}\alpha_{\text{IA}}} [\mathbf{u}_{\textrm{IA}}(t)]_n$ $ [\mathbf{u}_{\textrm{UI}}(t)]_n, n \in [2,N+1]$.
 The covariance matrix of $\mathbf{u}_g(t)$ is denoted as  $\mathbf{C}_g = \mathbb{E}\big[\mathbf{u}_g \mathbf{u}_g^H\big]$.
It is readily seen that although the  state equation \eqref{S3} is linear, the elements in $\mathbf{h}(t)$ do not follow complex Gaussian distributions in general, which implies that
the state variable $\mathbf{h}(t)$ and   noise variable $\mathbf{u}_g(t)$  are not Gaussian vectors. This is quite different from
 the state equation  \eqref{S1} discussed in the
special case, and  the KF-based channel tracking algorithm cannot be directly applied to this general case. Note that for the non-Gaussian scenario, the particle filter (PF) method is widely-used \cite{978374}, such as in target tracking, signal processing and automatic control, etc.   The PF method uses a set of particles, i.e., observations, to represent the posterior distribution of the state variable which follows an arbitrary distribution.  However, collecting  these particles can be quite time-consuming, thus applying the   PF method  for the general case can be very inefficient. }

To tackle this difficulty, we propose in this paper a GKF-based channel tracking algorithm for the general case.\footnote{The proposed GKF-based  channel tracking algorithm can be extended to handle other channel models as well as long as  the state   equation is linear.}
Our idea is inspired by the extended KF (EKF) and unscented KF (UKF) methods \cite{9044427}, both of which are designed for non-linear systems. Their difference is that the EKF method employs multivariate Taylor series expansions to linearize the state equation at the working point, while  the UKF method tries to approximate the probability density function of the output of the non-linear function  included in the system by a number of deterministic sampling  points which represent the underlying distribution as a Gaussian distribution. Then, the original KF method can work on these modified or linearized systems.
In the considered problem, our state equation and observation function are both linear, but the composite  noise variable $\mathbf{u}_g(t)$ in the state equation is not Gaussian.
Therefore,   some necessary modifications should be made on
 $\mathbf{u}_g(t)$  such that the resulting problem can be solved by  the KF method.
Specifically, to tackle the difficulty caused by  the non-Gaussian noise variable  $\mathbf{u}_g(t)$,
 we propose to use a complex Gaussian distribution to approximate the distribution of  $\mathbf{u}_g(t)$  in each time interval. Then, the predicted state variable $\mathbf{h}(t)$, obtained from the linear state equation \eqref{S3}, also follows a complex Gaussian distribution. With the help of such an approximation, we can apply the  KF method to address the resulting channel tracking problem.
In the following, we will first present the  proposed approximation method and then introduce the GKF-based channel tracking algorithm.
\vspace{-0.8cm}\subsection{Channel Distribution Analysis}\label{Sec41}

{
In this subsection, we aim to  analyse the distribution of the noise variable $\mathbf{u}_g(t)$ included in the state function \eqref{S3}.
First, it is readily seen that $[\mathbf{u}_g(t)]_1$ follows the complex Gaussian distribution, i.e., $[\mathbf{u}_g(t)]_1\sim \mathcal{CN}(0,\alpha_{\text{UA}}\beta_{\text{UA}})$. The other elements in $\mathbf{u}_g(t)$ can be expressed as the sum of three random variables, i.e.,
\vspace{-0.3cm}
\begin{equation}
[\mathbf{u}_{g}(t)]_n = [\mathbf{u}_{g,1}(t)]_n+[\mathbf{u}_{g,2}(t)]_n+[\mathbf{u}_{g,3}(t)]_n,
\end{equation}

\vspace{-0.2cm}\noindent
where $[\mathbf{u}_{g,1}(t)]_n \triangleq  \sqrt{\alpha_{\text{IA}}(1-\alpha_{\text{UI}})} [\mathbf{h}_r(t-1)(t)]_n [\mathbf{u}_{\textrm{IA}}(t)]_n$,
$[\mathbf{u}_{g,2}(t)]_n \triangleq  \sqrt{\alpha_{\text{UI}}(1-\alpha_{\text{IA}})} [\mathbf{g}(t-1)]_n[\mathbf{u}_{\textrm{UI}}(t)]_n$
and $ [\mathbf{u}_{g,3}(t)]_n \triangleq \sqrt{\alpha_{\text{IA}}\alpha_{\text{UI}}} [\mathbf{u}_{\textrm{IA}}(t)]_n [\mathbf{u}_{\textrm{UI}}(t)]_n$.
We can see that  $[\mathbf{u}_{g,1}(t)]_n$ and $[\mathbf{u}_{g,2}(t)]_n$ both follow complex Gaussian distributions, i.e.,
 $[\mathbf{u}_{g,1}(t)]_n \sim \mathcal{CN}(0, \alpha_{\text{IA}}(1-\alpha_{\text{UI}})\sigma^2_{\text{IA}}|[\mathbf{h}_r(t-1)]_n|^2)$, $ [\mathbf{u}_{g,2}(t)]_n \sim \mathcal{CN}(0, \alpha_{\text{UI}}(1-\alpha_{\text{IA}})|[\mathbf{g}_r(t-1)]_n|^2\sigma^2_{\text{UI}})$,
   while the third random variable  $[\mathbf{u}_{g,3}(t)]_n$ is  the product of two independent complex Gaussian random variables, i.e., $\sqrt{\alpha_{\text{IA}}} [\mathbf{u}_{\textrm{IA}}(t)]_n$ and $\sqrt{\alpha_{\text{UI}}} [\mathbf{u}_{\textrm{UI}}(t)]_n$.


  \begin{figure}[t]
\vspace{-0.0cm}
\setlength{\belowcaptionskip}{-0.5cm}
\renewcommand{\captionfont}{\small}
\centering
\includegraphics[scale=.52]{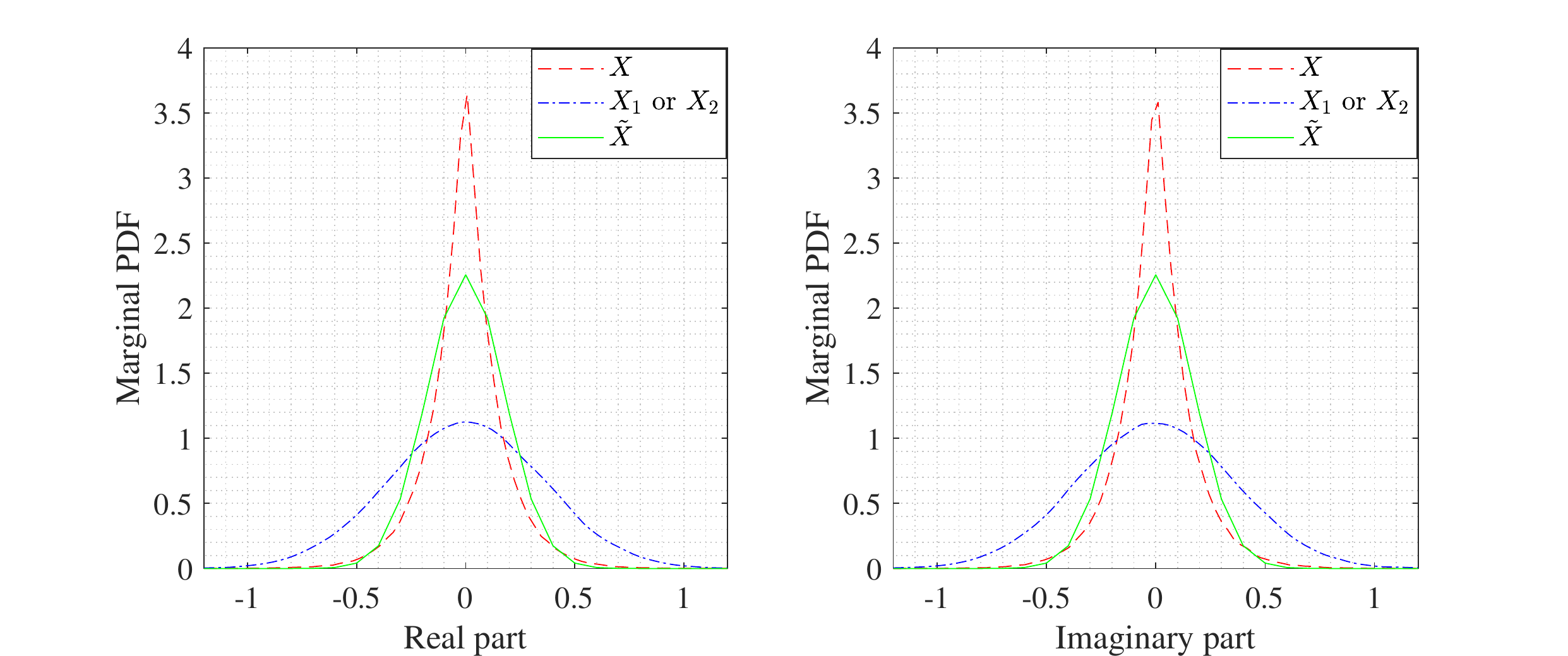}
\caption{Marginal probability distribution  of the real and imaginary  parts of $X,\{X_1,X_2\}$ and $\tilde{X}$. }
\label{MaPDF}
\normalsize
\end{figure}
Then, we focus on  the analysis of $[\mathbf{u}_{g,3}(t)]_n$.
For clarity, we refer to the distribution of this type of random variables as the product complex Gaussian distribution in the following.
Note that the product of two real Gaussian distributions, named as product real Gaussian distribution, has been  studied before in \cite{Springer1970,stojanac2018products}, where the authors proved that it can be  expressed in terms of  Meijer G-functions.
Consider a complex random variable $X$ that satisfies  $X = X_1X_2$,  where $X_1,X_2\sim \mathcal{CN}(0,\frac{1}{4})$.
Since it is difficult to directly analyse  the distribution of $X$, we first  investigate it via  Monte-Carlo simulation using $10^5$ samples.
{Fig \ref{MaPDF} demonstrates  the  marginal probability distributions of the  real and imaginary parts of $X$ and $X_1/X_2$.}
It is observed that the joint (or marginal) distribution resembles a shaper and slimmer complex Gaussian (or Gaussian) distribution, which inspires us to approximate
the product complex Gaussian distribution   by a simple complex Gaussian distribution with the same mean and variance.
However, since $[\mathbf{u}_{g,3}(t)]_n$ is in essence a component of the system noise variable $[\mathbf{u}_g(t)]_n$, we cannot sample this random variable  independently and calculate its mean and variance.
Hence, the ideas employed in the  PF and UKF methods, i.e., sampling a set of particles    to represent the posterior distribution of the state variable, cannot be  applied here to approximate the distribution of $[\mathbf{u}_{g,3}(t)]_n$.
{  To overcome this challenge, we resort to the following result about  the mean and variance of the random variable $X$, i.e., $\mu_P$ and $\sigma^2_P$,
\vspace{-0.3cm}\begin{equation}\label{theo1}
\begin{aligned}
&\mu_P = \mathbb{E}\big[X_1\big]\mathbb{E}\big[ X_2 \big] = 0,\\
&\sigma^2_P = \mathbb{E}\big[XX^*\big]=\mathbb{E}\big[X_1X_2X_2^*X_1^*\big]=\mathbb{E}\big[X_1X_1^*X_2X_2^*\big]=\mathbb{E}\big[X_1X_1^*\big]
\mathbb{E}\big[X_2X_2^*\big] = \sigma^2_1 \sigma^2_2.
\end{aligned}
\end{equation}

\vspace{-0.3cm}\noindent
As such, the mean and variance of $[\mathbf{u}_{g,3}(t)]_n$  can be directly  obtained  from  those of $\sqrt{\alpha_{\text{IA}}} [\mathbf{u}_{\textrm{IA}}(t)]_n$ and $\sqrt{\alpha_{\text{UI}}} [\mathbf{u}_{\textrm{UI}}(t)]_n$.
Then, based on  \eqref{theo1} and the numerical results shown in Fig. \ref{MaPDF}, a random product complex Gaussian variable $X$ with  mean $\mu_P$ and variance $\sigma^2_P$ can be approximated by a complex Gaussian variable $\tilde{X} \sim \mathcal{CN}(\mu_P, \sigma^2_P)$.}
As can be seen, although the distribution of  $\tilde{X}$ is similar to that of $X$, it is still not sharp enough as compared to  the true distribution.
 However, since only  one component in $[\mathbf{u}_{g}(t)]_n$, i.e., $[\mathbf{u}_{g,3}(t)]_n$, requires to be  approximated and the scaling coefficient of $[\mathbf{u}_{g,3}(t)]_n$, i.e., $\sqrt{\alpha_{\text{IA}}\alpha_{\text{UI}}}$, is  smaller than those of $[\mathbf{u}_{g,1}(t)]_n$ and  $[\mathbf{u}_{g,2}(t)]_n$, i.e.,   $\sqrt{\alpha_{\text{IA}}(1-\alpha_{\text{UI}})}$ and $\sqrt{(1-\alpha_{\text{IA}})\alpha_{\text{UI}}}$, we can infer that the negative effect  caused by the approximation error should be  limited. In Section \ref{Sec5}, we will show that the proposed GKF-based channel tracking algorithm
 can achieve good performance although under the employed approximation.
\vspace{-0.4cm}\subsection{Complex Gaussian Approximation}\label{Sec42}
  Next, we focus  on approximating  the distribution of $[\mathbf{u}_g(t)]_n$ based on the result in Section \ref{Sec41}.
According to the analysis in the previous subsection,  $[\mathbf{u}_{g,3}(t)]_n$ can be approximated by $[\hat{\mathbf{u}}_{g,3}(t)]_n$ which satisfies
$[\hat{\mathbf{u}}_{g,3}(t)]_n \sim \mathcal{CN}(0, \alpha_{\text{IA}}\alpha_{\text{UI}}\sigma^2_{\text{IA}}\sigma^2_{\text{UI}})$,
  thus $[\mathbf{u}_g(t)]_n$ can be approximated by  $[
\hat{\mathbf{u}}_g(t)]_n \triangleq [\mathbf{u}_{g,1}(t)]_n + [\mathbf{u}_{g,2}(t)]_n + [\hat{\mathbf{u}}_{g,3}(t)]_n, n\in [2,N+1]$.
It is noteworthy that the temporal correlation coefficients, i.e., $\alpha_{\text{IA}}$, $\alpha_{\text{UI}}$, $\alpha_{\text{UA}}$, and the variances of $[\mathbf{u}_{\text{IA}}]_n$, $[\mathbf{u}_{\text{UI}}]_n$ and $[\mathbf{u}_{\text{UA}}]_n$, i.e., $\sigma^2_{\text{IA}}$, $\sigma^2_{\text{UA}}$,  $\sigma^2_{\text{UI}}$, are   prior knowledges that are known. Since the sum of Gaussian distributions is still Gaussian,  the covariance matrix of $\mathbf{u}_g(t)$, i.e., $\mathbf{C}_g(t)$,  can be approximated by
\vspace{-0.3cm}\begin{equation}\label{Cg}
\hat{\mathbf{C}}_{g,\text{I}}(t) = \text{diag}\big(\alpha_{\text{UA}}\sigma^2_{\text{UA}}, \sigma^2_{g,1}(t),\cdots, \sigma^2_{g,N}(t) \big),
\end{equation}

\vspace{-0.2cm}\noindent
where
\vspace{-0.2cm}\begin{equation}
\sigma^2_{g,n}(t)  = \alpha_{\text{IA}}(1-\alpha_{\text{UI}})|[\mathbf{h}_r(t-1)]_n|^2\sigma^2_{\text{IA}} + \alpha_{\text{UI}}(1-\alpha_{\text{IA}})|[\mathbf{g}(t-1)]_n|^2\sigma^2_{\text{UI}}+ \alpha_{\text{IA}}\alpha_{\text{UI}}\sigma^2_{\text{IA}}\sigma^2_{\text{UI}}.
\end{equation}

\vspace{-0.2cm}\noindent
This is referred to as \emph{type one complex Gaussian approximation (CGA-I)} in the following.
It can be seen that the amplitudes of $[\mathbf{h}_r(t-1)]_n$ and $[\mathbf{g}(t-1)]_n$ are required in order  to calculate $\sigma^2_{g,n}(t)$, which is however difficult to obtain in the tracking process due to the passive nature of IRS, i.e., the CSI of the IRS-AP and user-IRS channels are difficult to obtain.   To address this issue, we propose to approximate $\alpha_{\text{IA}}(1-\alpha_{\text{UI}})|[\mathbf{h}_r(t-1)]_n|^2\sigma^2_{\text{IA}} +\alpha_{\text{UI}}(1-\alpha_{\text{IA}})|[\mathbf{g}(t-1)]_n|^2\sigma^2_{\text{UI}}$ using its lower bound, which is based on  the following theorem.
\vspace{-0.3cm}
\begin{theorem}\label{T2}
Assuming that $a = a_r + ja_i$ and $b = b_r + jb_i$ are two arbitrary complex scalars, then we have
\vspace{-0.2cm}\begin{equation}
|a|^2 + |b|^2 \ge \big|\Re(ab)\big| +  \big|\Im(ab)\big|,
\end{equation}
\end{theorem}

 \begin{proof}
Please relegate to Appendix.
\end{proof}

{Based on Theorem \ref{T2},  we employ  a random variable $
 [\tilde{\mathbf{u}}_g(t)]_n \sim \mathcal{CN}(0,\tilde{\sigma}^2_{g,n}(t))$ to approximate  $[{\mathbf{u}}_g(t)]_n$,
  where $\tilde{\sigma}^2_{g,n}(t)\le \text{Var}\big[ [{\mathbf{u}}_g(t)]_n \big]$}.\footnote{
    Note that in the KF method,  the  covariance matrix of the noise variable, i.e., $\mathbf{u}_g(t)$, is used to measure the estimation error of the state variable, i.e., $\mathbf{h}(t)$. Therefore,
  approximating $\text{Var}\big[ [{\mathbf{u}}_g(t)]_n \big]$ using the lower bound $\tilde{\sigma}^2_{g,n}(t)$ will make the estimated variable more stable. }
  Specifically, the value of $\tilde{\sigma}^2_{g,n}(t)$ is given by
 \vspace{-0.3cm}\begin{equation}
\tilde{\sigma}^2_{g,n}(t) = \delta_1
\Big(\big|\Re([\mathbf{h}(t-1)]_n)\big|   + \big|\Im([\mathbf{h}(t-1)]_n)\big|\Big) + \delta_2,
\end{equation}

\vspace{-0.2cm}\noindent
where
$\delta_1 \triangleq \sqrt{\alpha_{\text{IA}}(1-\alpha_{\text{IA}})\alpha_{\text{UI}}(1-\alpha_{\text{UI}})}
\sigma_{\text{IA}}\sigma_{\text{UI}}$ and
$\delta_2 \triangleq \alpha_{\text{IA}}\alpha_{\text{UI}}\sigma^2_{\text{IA}}\sigma^2_{\text{UI}}$.
Then, the covariance matrix $\mathbf{C}_{g}(t)$ can be approximated by
\vspace{-0.3cm}\begin{equation}\label{Cg1}
\hat{\mathbf{C}}_{g,\text{II}}(t) =
\left[\begin{array}{cc}
\alpha_{\text{UA}}\sigma^2_{\text{UA}} & \mathbf{0}_{1\times N}\\
\mathbf{0}_{N\times 1}   &
\delta_1\text{diag}\big(\big|\Re([\mathbf{h}(t-1)]_{2:N+1})\big|   + \big|\Im([\mathbf{h}(t-1)]_{2:N+1})\big|\big) +  \delta_2\mathbf{I}_N
\end{array}
\right],
\end{equation}

\vspace{-0.2cm}\noindent
where $[\mathbf{h}(t-1)]_{2:N+1} \triangleq \big[[\mathbf{h}_{2,N}(t-1)]_2,\cdots, [\mathbf{h}(t-1)]_{N+1} \big]^T \in \mathbb{C}^{N\times 1}$. This is named as \emph{type two complex Gaussian approximation (CGA-II)} in the following.
Note that CGA-I and CGA-II are different in the sense that the IRS-AP and user-IRS channel coefficients are required in CGA-I, while they are not required in CGA-II. However, the performance achieved by CGA-I and CGA-II is quite close, as will be shown in the simulation results.
 \vspace{-0.4cm}\subsection{GKF-based Channel Tracking}
 {
 Finally, based on the analysis presented in Section \ref{Sec41} and \ref{Sec42}, we extend the KF-based channel tracking algorithm (i.e., Algorithm \ref{Kalman}) to the general case and propose in Algorithm \ref{NonKalman} the GKF-based channel tracking algorithm.  Different from Algorithm \ref{Kalman} where only the Kalman gain $\mathbf{G}_{\text{KF}}(t)$, the correction $\mathbf{h}_{\text{KF}}(t)$ and the estimation covariance $\mathbf{M}_{\text{KF}}(t)$ are updated in each time interval, the covariance matrix $ \hat{\mathbf{C}}_{gl}(t)$ is also required to be updated in Algorithm \ref{NonKalman} since it varies with time in the general case.
  According to \eqref{Cg},  the estimation of $\hat{\mathbf{C}}_{gl}(t)$  can be obtained by
 \vspace{-0.3cm}\begin{equation}\label{Cge}
 \hat{\mathbf{C}}_{gl}(t) =\left[\begin{array}{cc}
\alpha_{\text{UA}}\sigma^2_{\text{UA}} & \mathbf{0}_{1\times N}\\
\mathbf{0}_{N\times 1}   &
\delta_1\text{diag}\big(\big|\Re([\mathbf{h}_{\text{KF}}(t-1)]_{2:N+1})\big|   + \big|\Im([\mathbf{h}_{\text{KF}}(t-1)]_{2:N+1})\big|\big) +  \delta_2\mathbf{I}_N
\end{array}
\right].
 \end{equation}
 \vspace{-0.5cm}

In the second stage of the proposed transmission protocol,    imaginary observations can be similarly obtained by the OB-LSTM network as in Section \ref{Sec3}, which serve as the input of Algorithm \ref{NonKalman}, then the channels in this stage can be estimated and no channel training overhead is required.}


 \begin{algorithm}[h]
\setlength{\belowcaptionskip}{-2cm}
\setlength{\abovecaptionskip}{-2cm}
\caption{GKF-based Channel Tracking Algorithm} 
\label{NonKalman}
\hspace*{0.02in} {\bf Input:}  covariance matrix  $\mathbf{C}_h = \mathbb{E}\big[\mathbf{h}\mathbf{h}^H\big]$ and observations {  $\tilde{\mathbf{y}}(t)$} \\
\hspace*{0.02in} {\bf Output:} {  Estimated channel vector $\mathbf{h}_{\text{KF}}(t)$}
\begin{algorithmic}[1]
\State \textbf{Initialization:} $t = 0$, $\mathbf{h}_{\text{KF}}(-1) = \mathbf{0}$, $ \mathbf{M}_{\text{MF}} = \mathbb{E}\big[\mathbf{h}\mathbf{h}^H\big]$, $\hat{\mathbf{C}}_{gl}(0) = \mathbb{E}\big[\mathbf{h}\mathbf{h}^H\big]$
\While {{  $t\ge 1$}}

\emph{Prediction step}:

\State Calculate the prediction according to \eqref{KF1},

 \State Calculate the prediction covariance matrix according to

$\mathbf{M}(t) = \mathbb{E}\big[\big(\mathbf{h}(t)-\hat{\mathbf{h}}(t) \big)\big( \mathbf{h}(t)-\hat{\mathbf{h}}(t)\big)^H  \big]
= \mathbf{A}_h\mathbf{M}_{\text{KF}}(t-1)\mathbf{A}_h^T + \mathbf{B}_h\hat{\mathbf{C}}_{gl}(t)\mathbf{B}_h^T$,

 \emph{Update step}:

\State Update the Kalman gain vector according to \eqref{KF3},

\State Update the correction according to \eqref{KF4},

\State Update the estimation covariance matrix according to \eqref{KF5},

\State Update the approximated covariance matrix of $\mathbf{u}_g$, $\hat{\mathbf{C}}_{gl}(t+1)$, according to \eqref{Cge},
\State $t = t+1$,
\EndWhile
\end{algorithmic}
\end{algorithm}

%

\vspace{-0.8cm}
\section{Simulation Results}\label{Sec5}

{In this section, we present numerical results to validate the effectiveness of the proposed two-stage transmission protocol and the 2SCTP scheme. We consider an IRS-aided wireless communication system as shown in Fig. \ref{system}, where
  a three-dimensional coordinate system  is assumed and  the AP  and   IRS are located
on the $x$-axis and $y-z$ plane, respectively. The reference antenna/element at the
AP/IRS  are located at $(3$ $\text{m}, 0, 0)$ and $(0, 50$ $\text{m}, 2$ $\text{m})$,  and the user is located at  $(2$ $\text{m}, 50$ $\text{m}, 0)$.
In the simulations, the IRS is equipped with a $5\times 7$ uniform rectangular
array.\footnote{If the number of IRS reflecting elements is large, the channel training overhead required to obtain the observations and the number of network parameters included in the OB-LSTM network will also be increased. In this case, we can employ the grouping and partition method \cite{9133142}  to achieve a good trade-off between channel tracking performance and channel training overhead/network complexity.} The reference distance and the pass loss at the reference distance are set as $d_0 = 1$ m and $l_0=-30$ dB. The IRS-AP, user-IRS and user-AP path-loss exponents are fixed to $ \gamma^{\text{IA}} = 2.2$, $\gamma^{\text{UI}} = 2.2$ and $\gamma^{\text{UA}}=3.6$, respectively. The transmit power and noise variance are set to $p = 26$ dBm and $\sigma^2 = -80$ dBm.
{  Furthermore, the variances of the elements of $\mathbf{u}_{\text{IA}}$, $\mathbf{u}_{\text{UI}}$ and  ${u}_{\text{UA}}$ are respectively set to $\sigma^2_{\text{IA}} = l_{\text{IA}}$, $\sigma^2_{\text{UA}}  =  l_{\text{UA}}$ and $\sigma^2_{\text{UI}}   = l_{\text{UI}}$.}
In all our simulations, the definitions of the normalized mean square error (NMSE) and average NMSE (ANMSE) in the  $t$-th time interval  are given by
$
\text{NMSE}(t) = \frac{\|\mathbf{h}_{\text{KF}}(t) - {\mathbf{h}}(t)  \|^2}{\|\mathbf{h}(t) \|^2}, \;\;\;\;
\text{ANMSE}(t) = \frac{1}{t}\sum_{i=1}^t\frac{\|\mathbf{h}_{\text{KF}}(i) - {\mathbf{h}}(i)  \|^2}{\|\mathbf{h}(i) \|^2}.
$


\begin{figure}[h]
\vspace{-0.3cm}
\setlength{\belowcaptionskip}{-0.8cm}
\renewcommand{\captionfont}{\small}
\centering
\includegraphics[scale=.30]{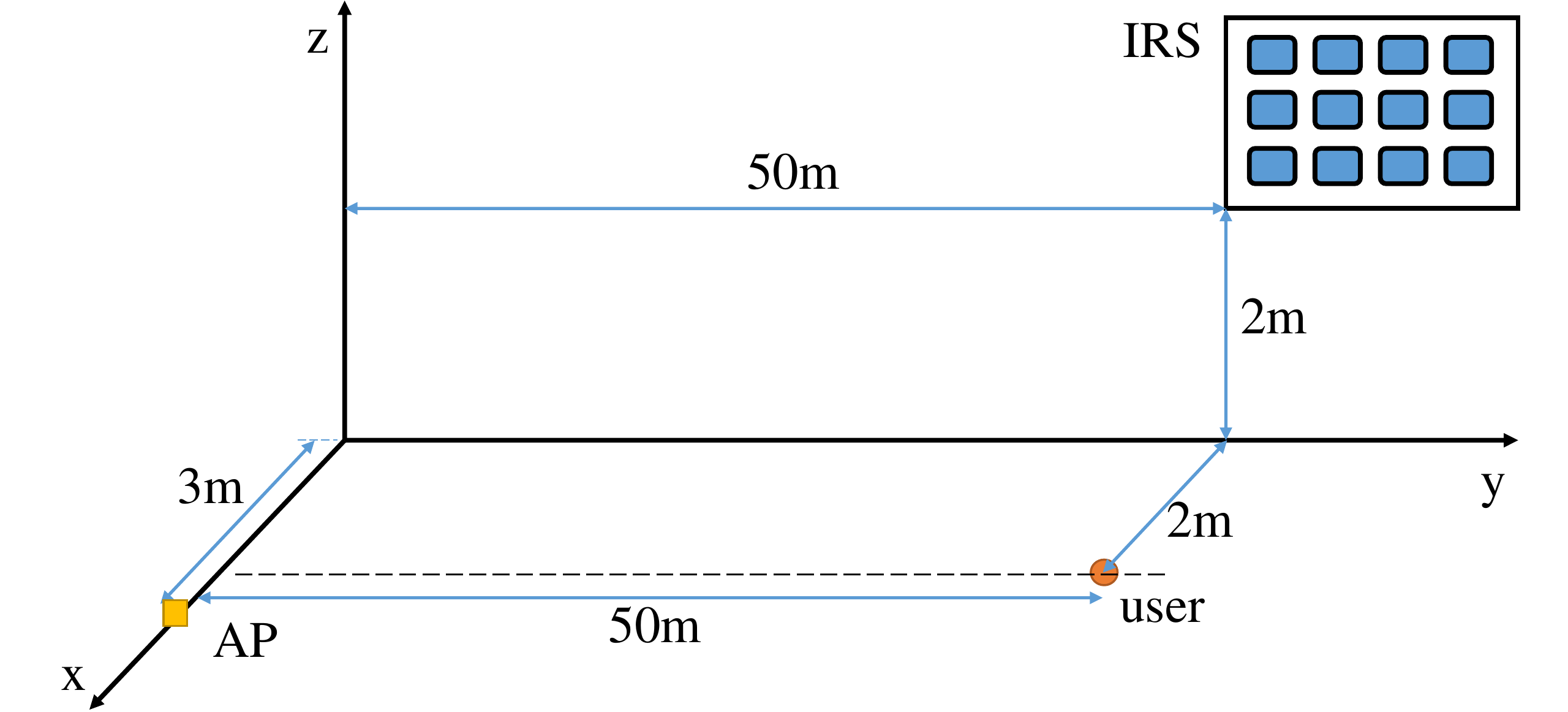}
\caption{Simulation setup.}
\label{system}
\normalsize
\end{figure}
\vspace{-0.4cm}
\subsection{Special Case}
In this subsection, we consider the special case where   the IRS-AP channel is assumed to stay unchanged, i.e., $\alpha^{\text{IA}}=0$, and the temporal correlation coefficients of the
user-IRS and user-AP channels are set as  $\alpha^{\text{UI}} = 0.01$ and $\alpha^{\text{UA}} = 0.01$, respectively.

First, we investigate the channel tracking performance in the first stage.
Fig. \ref{NMSE-tau} shows
 the NMSE performance of the KF-based channel tracking algorithm (i.e., Algorithm \ref{Kalman}) with different numbers of time intervals. {  The   number of time slots in each time interval for transmitting pilot symbols is assumed to be $\tau_1 \in \{2,6,10\}$.}
For comparison, we also provide the performance achieved by the channel estimation (CE)  scheme proposed in \cite{Jensen2019} as the benchmark, i.e., in each time interval, $N+1$ time slots are allocated to transmit pilot symbols and the channels are estimated by using the DFT reflection pattern and the MMSE channel estimation algorithm.
As can be seen, the performance of  Algorithm \ref{Kalman} improves with the increasing of $\tau_1$ and in the meantime,  higher convergence speed can be achieved.
For example,  when $\tau_1 = 10$, the number of time intervals required by Algorithm \ref{Kalman} to achieve convergence is less than 10,    while when $\tau_1 = 2$, this number increases to 20.
Moreover, with larger $\tau_1$, Algorithm \ref{Kalman} shows  more stable NMSE performance after   convergence, e.g., the fluctuation of the NMSE curve  when $\tau_1 = 10$ is smaller  than that when $\tau_1 = 2$.
Furthermore,  as compared to the benchmark, Algorithm \ref{Kalman} can achieve much lower (150 time slots versus 525 time slots if 15 time intervals and $\tau_1 = 10$ are considered) channel training overhead with only minor NMSE performance loss.
%

\begin{figure}[b]
\vspace{-1cm}
\setlength{\belowcaptionskip}{-1cm}
\renewcommand{\captionfont}{\small}
\begin{minipage}[t]{0.5\textwidth}
\centering
\includegraphics[scale=.45]{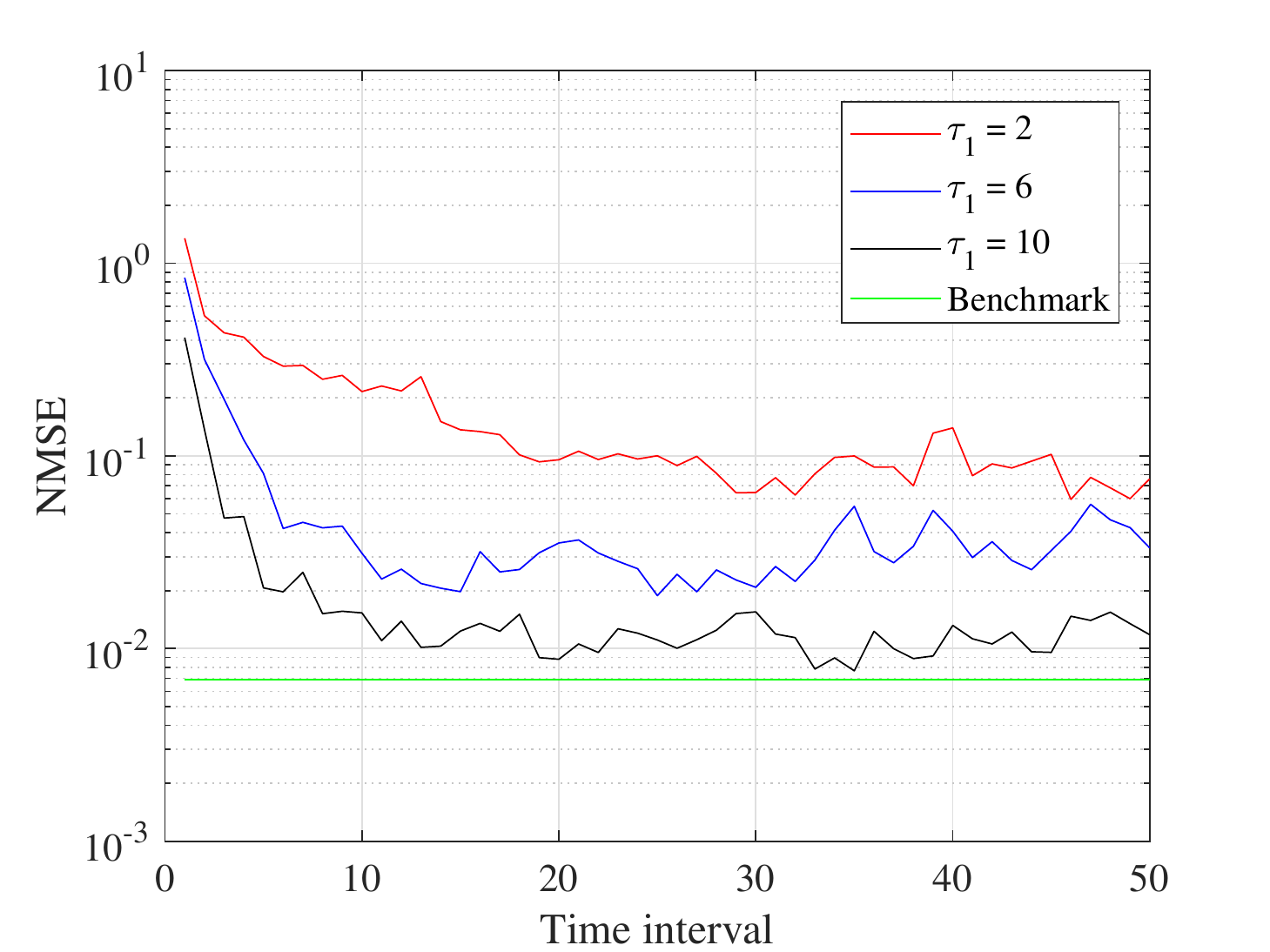}
\caption{Channel tracking performance in terms of  NMSE when different $\tau_1$ are assumed in the first stage (special case). }
\label{NMSE-tau}
\normalsize
\end{minipage}
\;
\begin{minipage}[t]{0.5\textwidth}
\centering
\includegraphics[scale=.45]{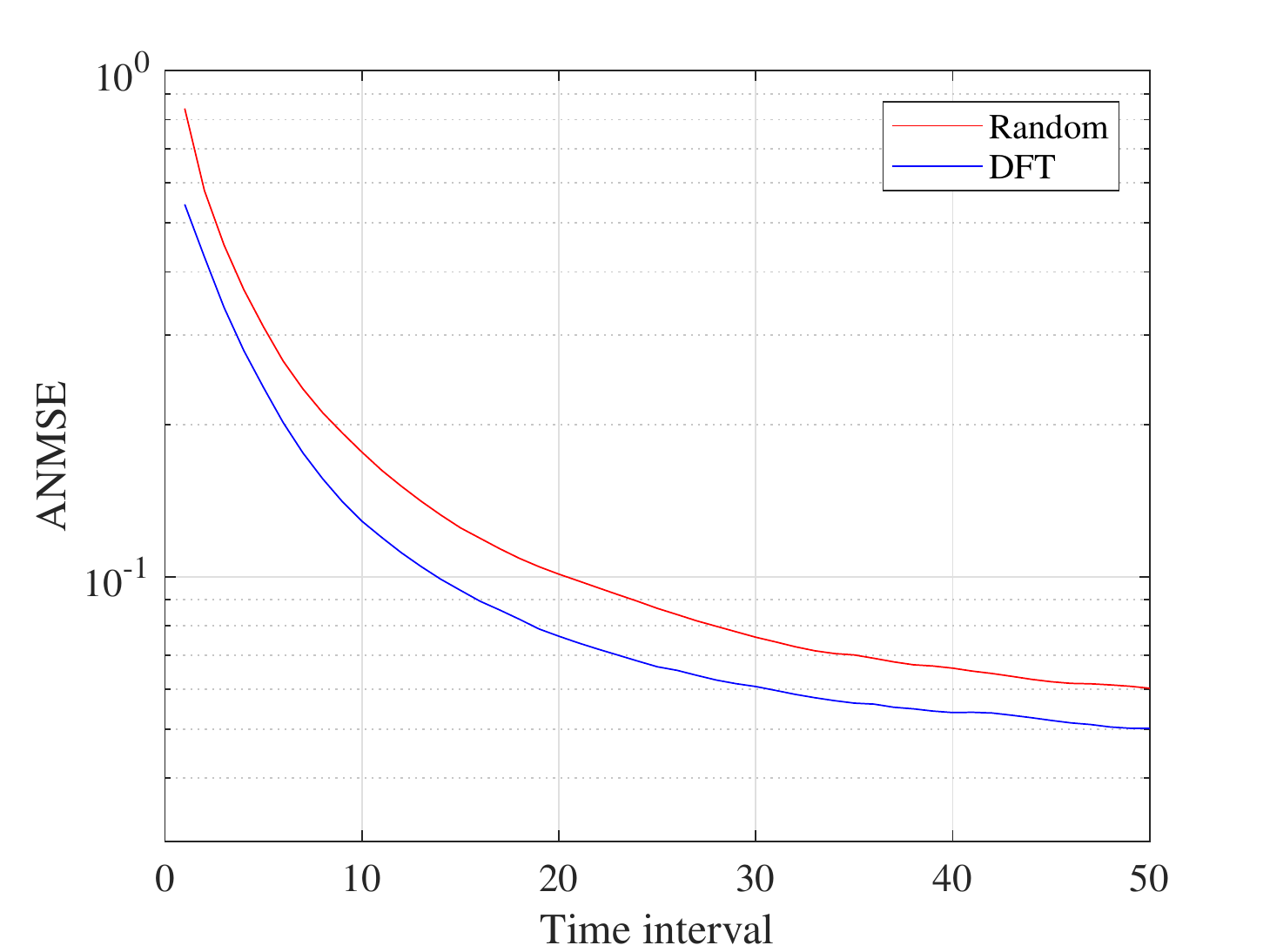}
\caption{Channel tracking performance in terms of ANMSE when two different {measurement matrices} are considered (special case). }
\label{ANMSE-DFT}
\normalsize
\end{minipage}
\end{figure}


In Fig. \ref{ANMSE-DFT}, we investigate    the ANMSE performance of Algorithm \ref{Kalman} in the first stage  when two different   {  measurement matrices $\mathbf{V}(t)$} are employed, i.e.,  {  the random measurement matrix and the DFT measurement matrix}.
{  In the random measurement matrix, each element of the reference matrix $\mathbf{Q}$ in \eqref{Q} is generated from a complex Gaussian distribution $\mathcal{CN}(0,1)$, while in the DFT measurement matrix,  $\mathbf{Q}$ is an $(N+1)$-DFT matrix. $\tau_1$ is fixed to 6.}
One can see that as compared to using   random {measurement matrix}, Algorithm \ref{Kalman} with the DFT {  measurement matrix}  exhibits better convergence. This is due to the fact that all the columns of the DFT {  measurement matrix} are independent  with each other, which forms a better sensing matrix.



Second, we investigate the channel prediction performance in the second stage.   In our simulations, the number of LSTM layers in the OB-LSTM network is set to  $K = 4$. The DFT {  measurement matrix} is used  due to its superiority demonstrated in Fig. \ref{ANMSE-DFT}. The number of training  samples is set to $10^4$, and the  batchsize and   learning rate are fixed to 10 and 0.0001, respectively. The training process is conducted
offline on a Windows server with Intel Xeon
Gold 6230 CPU and an Nvidia 2080Ti GPU, and the proposed networks are  implemented in Python using
the TensorFlow library with the Adam optimizer.

We present in Fig. \ref{convergence55}   the convergence of the proposed OB-LSTM network during   training, {  where the $\mathcal{L}_2$ loss in \eqref{L2}
  when feeding the validation data set into the OB-LSTM network is regarded as the   performance metric.}
We consider four different configurations of the network structure, i.e., the scaling factor of the LSTM units $\epsilon$ is set to be varying  in $[1,3,5,10]$. Besides, two different configurations of the lengths of the input  observations and predicted imaginary observations are considered: (a) $L_I = 6, L_P = 1$, and (b) $L_I = 6, L_P = 6$.
{  From Fig. \ref{convergence55}, we can observe that    the $\mathcal{L}_2$ loss achieved by the OB-LSTM network is able to converge in about 20000 iterations.}
Besides,   the performance of the OB-LSTM network improves  with the increasing of the scaling factor $\epsilon$. However, larger  $\epsilon$ also means higher computational complexity. Therefore, the hyper-parameter $\epsilon$ should be carefully chosen  before training to strike a good balance between performance and computational complexity.
Furthermore, by comparing Fig. \ref{convergence55} (a) and (b), we can see that with the same value of $\epsilon$, the OB-LSTM network in case (a) outperforms that in case (b).
 This implies that the prediction performance of the OB-LSTM network deteriorates as the length of the predicted imaginary observations $L_P$ increases, and there is a tradeoff between  channel training overhead (i.e., prediction length) and  prediction performance.

\begin{figure}[t]
\vspace{-0.5cm}
\setlength{\belowcaptionskip}{-0.5cm}
\renewcommand{\captionfont}{\small}
\centering
\includegraphics[scale=.45]{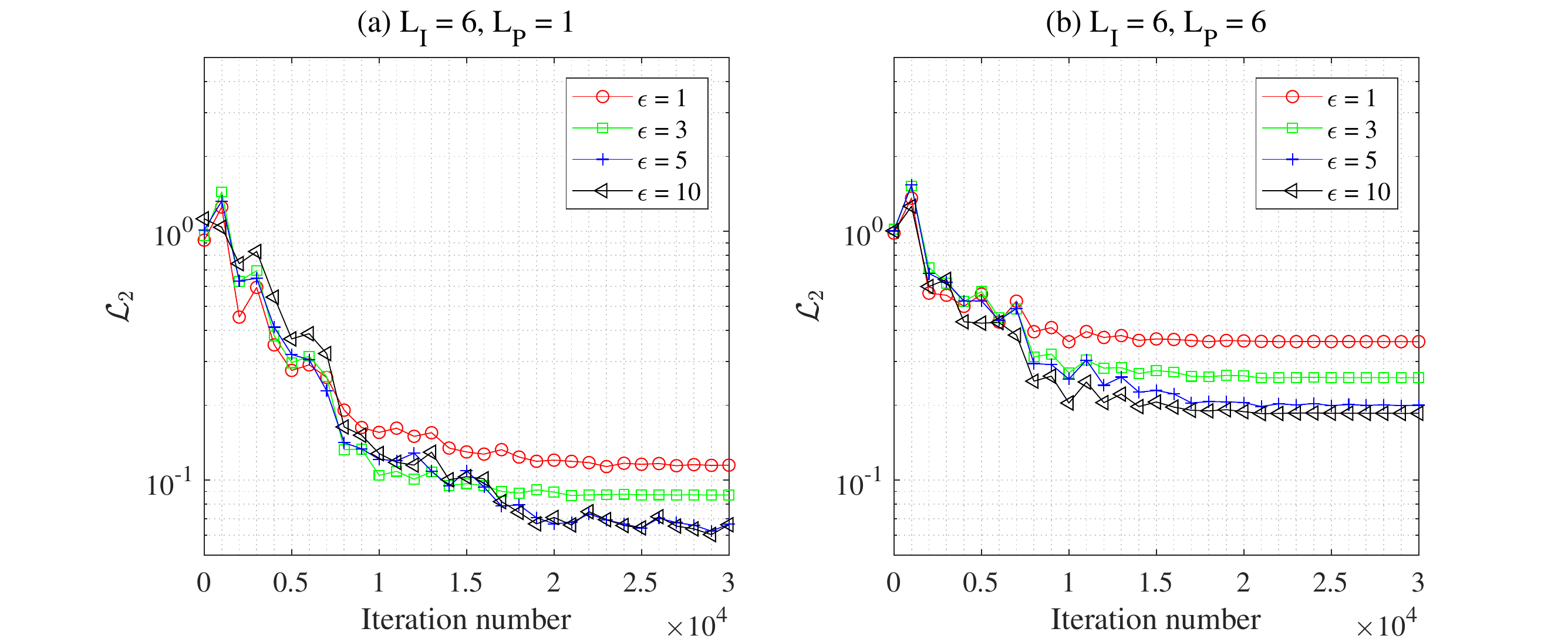}
\caption{Convergence of the proposed OB-LSTM network during training. }
\label{convergence55}
\normalsize
\end{figure}

Next, we exhibit in TABLE \ref{Tab1}   {  the NMSE of the predicted imaginary observations (OB-NMSE)} achieved by the proposed OB-LSTM network, when the number of time intervals allocated for the two stages  are fixed to $T_1 = 6$ and $T_2 = 6$, respectively. {We consider \emph{Strategy A} and \emph{Strategy B} discussed in Section \ref{OB} which meet the requirement of $T_1 = 6$ and $T_2 = 6$.}
 As shown in TABLE \ref{Tab1}, the {OB-NMSE} achieved by \emph{Strategy A} is quite stable over  all the considered time intervals,  while that by  \emph{Strategy B}  gradually deteriorates over time. Besides, \emph{Strategy B} outperforms \emph{Strategy A}
in the first three time intervals, while   \emph{Strategy A} shows better   performance in  the other time intervals.
{The former is due to the non-ideal processing introduced by the OB-LSTM network  when the output dimension $L_P$ is larger ($L_P=6$) in \emph{Strategy A}.
The latter is because  the  OB-LSTM network  in \emph{Strategy B} outputs the following imaginary observations one by one based on both real and imaginary observations,  as such the prediction  errors in the previous time intervals will deteriorate
the prediction performance in the subsequent time intervals  and thus \emph{Strategy B} becomes worse than \emph{Strategy A}
in the later time intervals. }

\begin{table*}[t]
\vspace{-0em}
\setlength{\abovecaptionskip}{-0.0cm}
\setlength{\belowcaptionskip}{-0.0cm}
\caption{Observation prediction performance in the second stage.}
\begin{center}
\begin{tabular}{|c|c|c|c|c|c|c|}
\hline
\multirow{2}{*} {  OB-NMSE}& \multicolumn{6}{c|}{Time interval index in the second stage} \\
\cline{2-7}
 & 1 & 2 & 3 & 4 & 5 & 6 \\
\hline
\emph{Strategy A} &  0.1058 & 0.1263 & 0.1208 & 0.1098 & 0.1026  & 0.1541\\
\hline
\emph{Strategy B} &  0.0297 & 0.0579 & 0.0748 & 0.1161 &  0.1235 & 0.1449 \\
\hline
\end{tabular}
\label{Tab1}
\end{center}
\vspace{-1.0cm}
\end{table*}

\begin{figure}[b]
\vspace{-1.0cm}
\setlength{\belowcaptionskip}{-0.0cm}
\renewcommand{\captionfont}{\small}
\centering
\includegraphics[scale=.48]{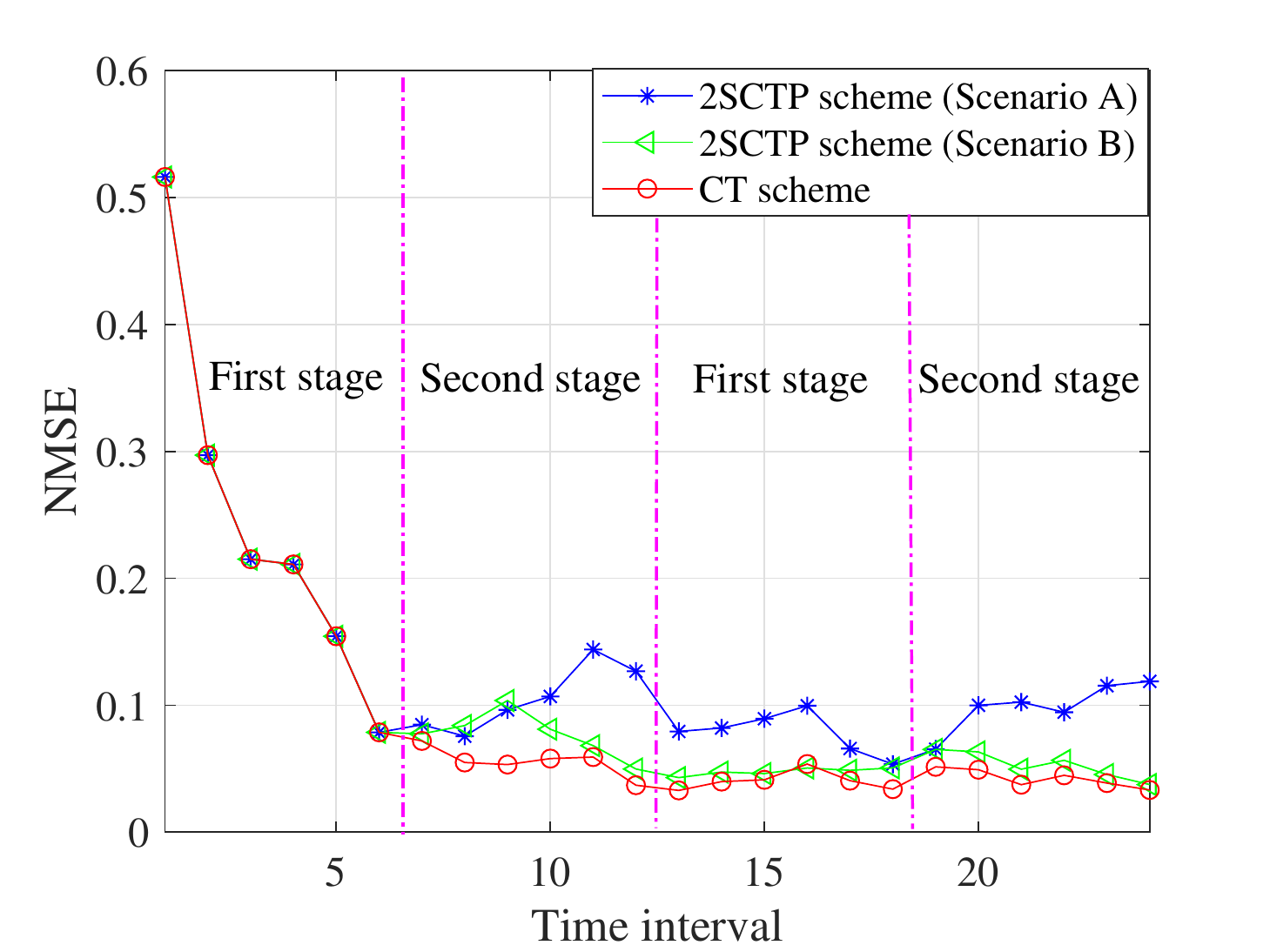}
\caption{NMSE performance of the proposed 2SCTP scheme in different scenarios. }
\label{SCTPper}
\normalsize
\end{figure}


Then, in Fig. \ref{SCTPper}, we investigate    the overall NMSE performance of the proposed 2SCTP scheme in two different scenarios, i.e., Scenario A and Scenario B.
 Specifically, in Scenario A, we set $T_1= 6$, $T_2 = 6$, $L_I = 6$, $L_P = 6$ and \emph{Strategy A} is employed to generate the predicted imaginary observations;
   while in Scenario B, we  set  $T_1= 6$, $T_2 = 3$, $L_I = 6$, $L_P = 1$ and \emph{Strategy B} is employed.
  For comparison, we consider a pure  channel tracking (CT) scheme, where
   there is no channel prediction, i.e., all the time intervals are allocated to the first stage.
Note that  the ratio  $\frac{T_2}{T_1}$  of Scenario A is larger than that of Scenario B, thus  lower channel training overhead can be achieved  by Scenario A.
 As shown in Fig. \ref{SCTPper},  in the first stage, the NMSE performance of the 2SCTP scheme in Scenario B is the same to that in Scenario A, however in the second stage,
 the performance  in Scenario B
 is better than that in Scenario A and it is quite close to
 that of the benchmark.
Therefore, for the proposed  2SCTP scheme, adopting the parameters in Scenario B can effectively reduce the channel training overhead with almost no  performance loss.

At last, {Table \ref{Tab2} compares the channel training overhead required by the 2SCTP scheme in  Scenarios A and   B when the number of  time intervals  is fixed to $T=3600$.   The   amount of channel training overhead is measured by  the number of time slots allocated for pilot transmission.
The overheads required by   the CE   and   CT schemes    are regarded as   benchmarks.
It can be observed from Table \ref{Tab2} that in Scenario B, the 2SCTP scheme  requires less channel training overhead than that in Scenario A, however the NMSE performance is worse in Scenario B (see Fig. \ref{SCTPper}).
  Moreover, the proposed 2SCTP scheme is able to  reduce the channel training overhead significantly as compared to the
  CE and CT schemes. In particular, the overhead required in Scenario B is only 8.3\% and 50.0\% of those by the CE and CT schemes, respectively.

  .
}
\begin{table*}[t]
\vspace{-0.0em}
\setlength{\abovecaptionskip}{-0.0cm}
\setlength{\belowcaptionskip}{-0.0cm}
\caption{Channel Training Overhead Comparison.}
\begin{center}
\begin{tabular}{|c|c|c|c|}
\hline
Scheme & 2SCTP scheme (Scenario A) & 2SCTP scheme (Scenario B) \\
\hline
Channel training overhead & $1.44\times 10^4$  & $1.08\times 10^4$ \\ 
\hline
Scheme & CT scheme & CE scheme \\
\hline
Channel training overhead & $2.16\times 10^4$  & $1.30\times 10^5$   \\
\hline
\end{tabular}
\label{Tab2}
\end{center}
\vspace{-1.2cm}
\end{table*}
\vspace{-0.6cm}
\subsection{General Case}
\begin{figure}[b]
\vspace{-1.2cm}
\setlength{\belowcaptionskip}{-1cm}
\renewcommand{\captionfont}{\small}
\begin{minipage}[t]{0.5\textwidth}
\centering
\includegraphics[scale=.45]{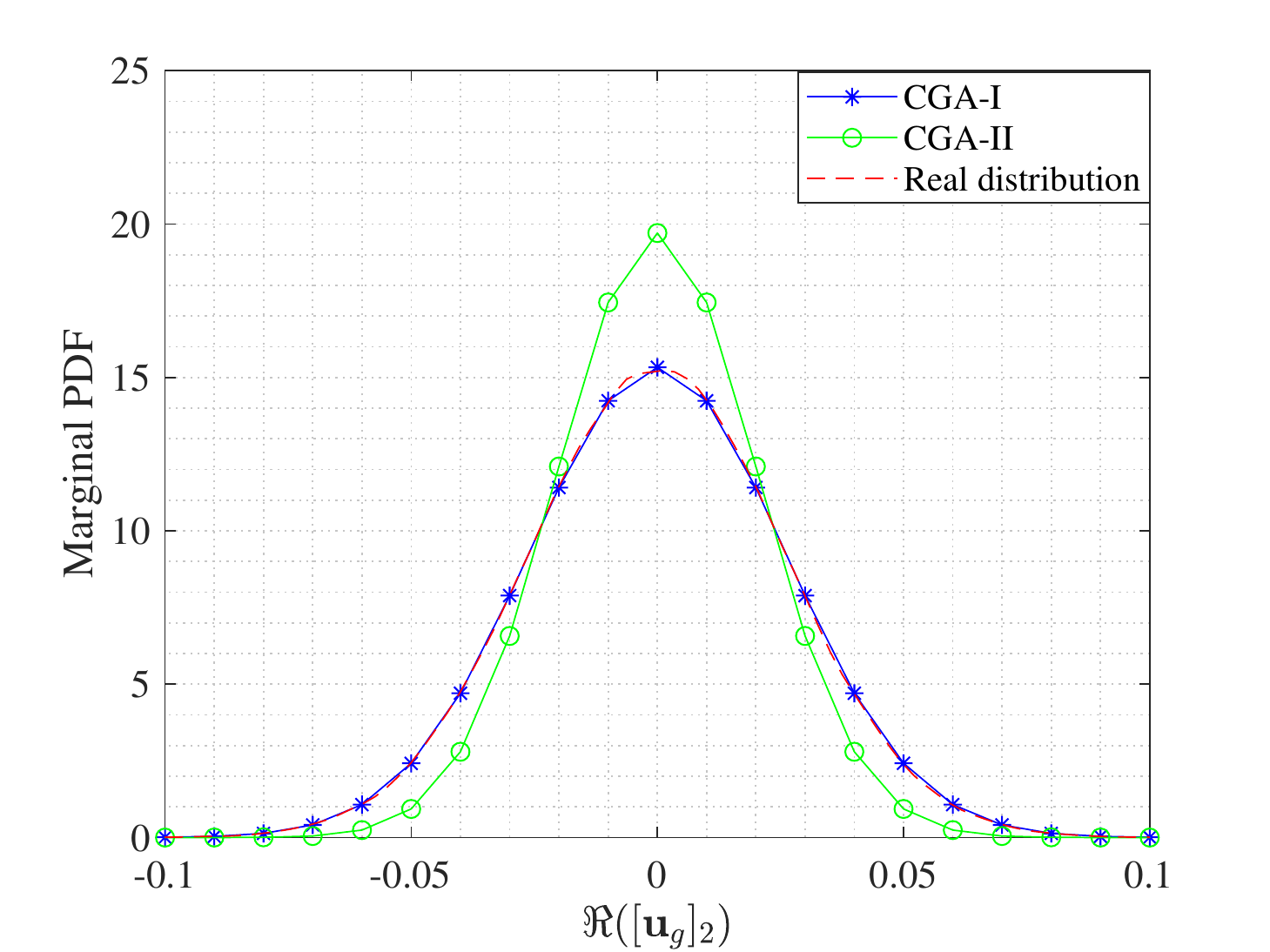}
\caption{Marginal distributions of the random variable $\Re([\mathbf{u}_g]_2)$. }
\label{GAU}
\normalsize
\end{minipage}
\;
\begin{minipage}[t]{0.5\textwidth}
\centering
\includegraphics[scale=.45]{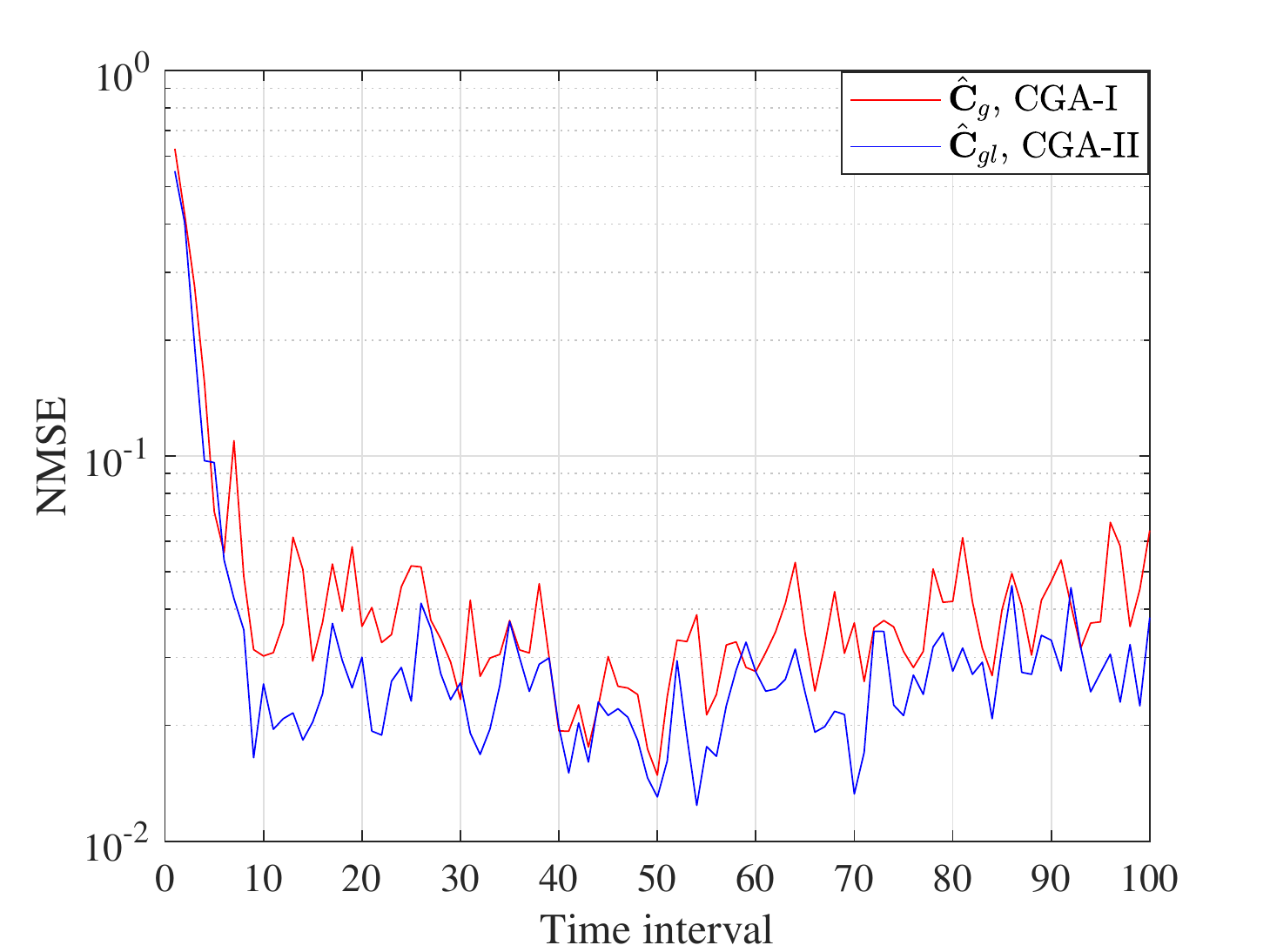}
\caption{Channel tracking process in terms of NMSE
where different approximated covariance matrices of $\mathbf{u}_g$ are assumed (general case).}
\label{NMSE-gen-a}
\normalsize
\end{minipage}
\end{figure}

In this subsection, we consider the general case and  set the   temporal correlation coefficients as $\alpha^{\text{IA}} = \alpha^{\text{UI}} = \alpha^{\text{UA}} = 0.01$.

First, we show in  Fig. \ref{GAU}   the real distribution of $\mathbf{u}_g$ and its approximations using  CGA-I in \eqref{Cg} and CGA-II in \eqref{Cg1}. For simplicity, we take the marginal distribution of the real and imaginary  parts of   $[\mathbf{u}_g]_2$ as an example. Besides, since the marginal distributions of $\Re([\mathbf{u}_g]_2)$ and $\Im([\mathbf{u}_g]_2)$ are the same, we only plot the distribution of $\Re([\mathbf{u}_g]_2)$ and its approximation in Fig. \ref{GAU}. It is observed that the distribution achieved by  CGA-I
 is almost the same with the real distribution, while the distribution achieved by  CGA-II has a smaller variance than the real distribution.



Then, Fig. \ref{NMSE-gen-a} shows  the  performance achieved by the proposed
 GKF-based channel tracking algorithm (i.e., Algorithm \ref{NonKalman}), when  two different approximated covariance matrices of $\mathbf{u}_g$  are used, i.e., $\hat{\mathbf{C}}_{g,\text{I}}$ given in \eqref{Cg} and $\hat{\mathbf{C}}_{gl}$ given in \eqref{Cge}, which are obtained via CGA-I and CGA-II, respectively.
It is noteworthy that $\mathbf{h}_r(t-1)$ and $\mathbf{g}(t-1)$  are  needed to calculate $\hat{\mathbf{C}}_{g,\text{I}}(t)$, while only $\mathbf{h}_{\text{KF}}(t-1)$ is required to obtain  $\hat{\mathbf{C}}_{gl}(t)$.
To employ $\hat{\mathbf{C}}_{g,\text{I}}$   in Algorithm  \ref{NonKalman}, we replace \eqref{Cge} with \eqref{Cg} in step 8 to update  the approximated covariance matrix of $\mathbf{u}_g$, and $\mathbf{h}_r(t-1)$ and $\mathbf{g}(t-1)$ are assumed to be known.\footnote{Due to the passive nature  of IRS, the IRS-AP channel $\mathbf{g}$ and the user-IRS channel $\mathbf{h}_r$ are difficult to be estimated. Here, to show the performance  of the proposed CGA-I method, we assume that
 $\mathbf{g}$  and $\mathbf{h}_r$ are known.} The   number of time slots allocated for the channel training phase is set to $\tau_1=6$. As shown in Fig. \ref{NMSE-gen-a},  by using $\hat{\mathbf{C}}_{gl}$ and $\hat{\mathbf{C}}_{g,\text{I}}$, Algorithm \ref{NonKalman} can achieve  similar NMSE performance, and the converged NMSE can reach  0.05  after about 10 time intervals.


Next, in Fig. \ref{NMSE-tau-g}, we investigate   the NMSE performance achieved by Algorithm \ref{NonKalman} with different values of $\tau_1$, i.e., $\tau_1 \in [2, 4, 6, 8]$. Similar to Fig. \ref{NMSE-tau}, we regard the NMSE performance of the CE scheme as the   benchmark. It can be observed that when $\tau_1$ is small, i.e., $\tau_1 = 2$, Algorithm \ref{NonKalman} almost fails to track the varying channel.
When $\tau_1 = 4, 6, 8$, the NMSE performance
achieved by Algorithm \ref{NonKalman}  can achieve convergence within $10$ time intervals. Besides, as $\tau_1$ increases, the convergence will become  faster and more stable. However, there is a performance gap (though not large) between Algorithm \ref{NonKalman} and the benchmark, which is because the number of   pilots required by  Algorithm \ref{NonKalman} is much less than that of the CE scheme.
\begin{figure}[t]
\vspace{-0.5cm}
\setlength{\belowcaptionskip}{-0.5cm}
\renewcommand{\captionfont}{\small}
\begin{minipage}[t]{0.5\textwidth}
\centering
\includegraphics[scale=.45]{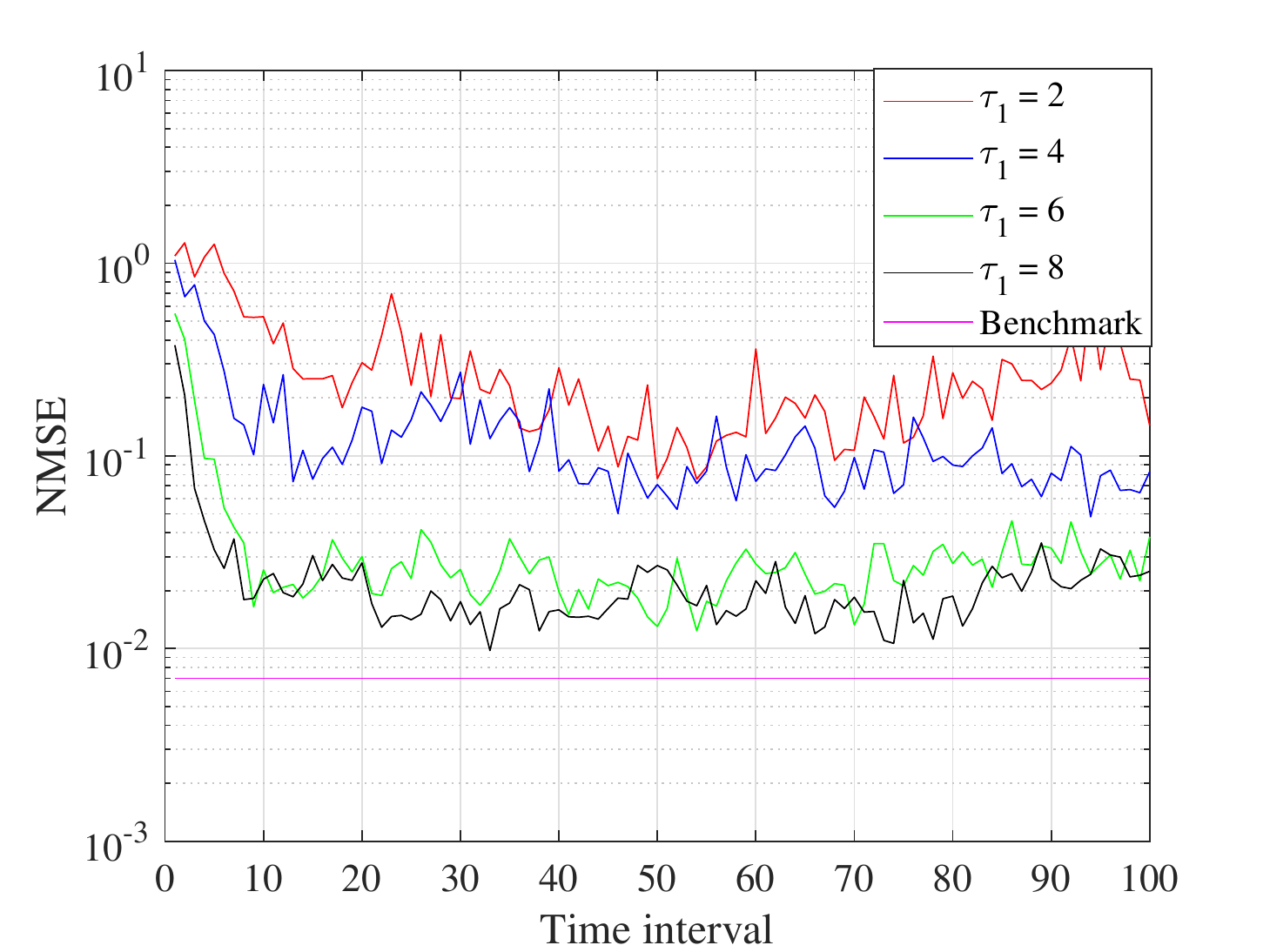}
\caption{Channel tracking process in terms of  NMSE when different $\tau_1$ are assumed in the first stage (general case). }
\label{NMSE-tau-g}
\normalsize
\end{minipage}
\;
\begin{minipage}[t]{0.5\textwidth}
\centering
\includegraphics[scale=.45]{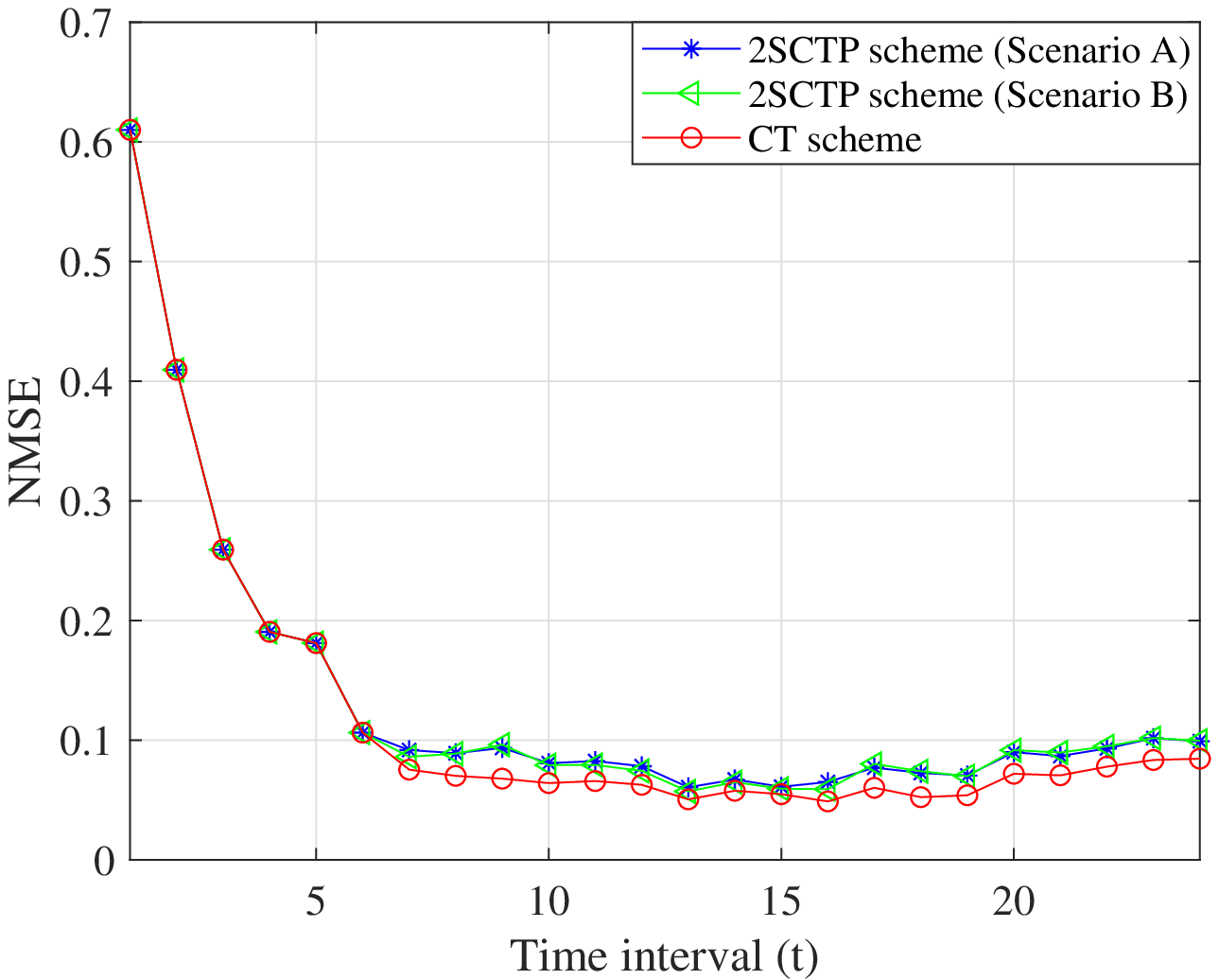}
\caption{NMSE performance of the proposed 2SCTP scheme in different scenarios (the general case). }
\label{SCTPperG}
\normalsize
\end{minipage}
\end{figure}


Finally,  Fig. \ref{SCTPperG} presents the NMSE performance of the proposed 2SCTP scheme in the general case by combining  Algorithm \ref{NonKalman} and the OB-LSTM network. Different from that in Fig. \ref{SCTPper}, we consider two scenarios with the same amount of channel training overhead, i.e., in Scenario A, we set $T_1= 6$, $T_2 = 3$, $L_I = 6$, $L_P = 3$ and employ \emph{Strategy A} to generate the predicted imaginary observations;
  while  in Scenario B, we  set $T_1= 6$, $T_2 = 3$, $L_I = 6$, $L_P = 1$ and use  \emph{Strategy B}. It can be seen that in the general case, the proposed 2SCTP scheme can also achieve  very similar NMSE performance with the CT scheme, yet with much lower channel training overhead.}

\vspace{-0.5cm}\section{Conclusion}
In this paper, we investigated the CTP problem in an  IRS-aided wireless communication system with time-varying channel and designed an innovative two-stage transmission protocol.
 Based on the proposed transmission protocol, we proposed a novel 2SCTP scheme to track and predict the channels with low channel training overhead.
 By exploiting the temporal correlation of the channels, we proposed a  KF-based channel tracking algorithm for the special case when the IRS-AP channel is static, while for the general case, we developed a GKF-based channel tracking algorithm by devising a  simple Gaussian approximation method. Furthermore, we presented an LSTM-based neural network, namely the \emph{OB-LSTM} network, to predict the imaginary observations based on which the channels can be estimated  by applying   the KF/GKF-based channel tracking algorithms.
Numerical results
showed that the  proposed 2SCTP scheme is able to  outperform the existing CE and CT schemes significantly in terms of channel training overhead.

\vspace{-0.5cm}
\begin{appendix}
\begin{proof}
First, the quadratic sum of $a$ and $b$, i.e.,  $|a|^2 + |b|^2$, can be rewritten as
\vspace{-0.2cm}\begin{equation}
|a|^2 + |b|^2 =   \frac{1}{2}\big|(a_r^2 +  b_r^2) + (a_i^2 +b_i^2)\big| + \frac{1}{2}\big|(a_r^2 + b_i^2) + (a_i^2 + b_r^2)  \big|.
\end{equation}
According to the arithmetic and geometric (AM-GM) inequality \cite{Steele2004}, $x^2 + y^2 \ge 2xy$ and  $x^2+y^2\ge -2xy$ hold for all $x,y\in \mathbb{R}$. As such, we have
\vspace{-0.2cm}\begin{equation}
\begin{aligned}
&\left.\begin{aligned}
&(a_r^2 +  b_r^2) + (a_i^2 +b_i^2) \ge 2a_rb_r-2a_ib_i\\
& (a_r^2 +  b_r^2) + (a_i^2 +b_i^2) \ge -2a_rb_r+2a_ib_i
 \end{aligned}\right\} \Rightarrow   \big|(a_r^2 +  b_r^2) + (a_i^2 +b_i^2)\big| \ge 2\big| a_rb_r-a_ib_i \big|, \\
&\left. \begin{aligned}
& (a_r^2 +  b_i^2) + (a_i^2 +b_r^2) \ge 2a_rb_i+2a_ib_r \\
& (a_r^2 +  b_i^2) + (a_i^2 +b_r^2) \ge -2a_rb_i-2a_ib_r
 \end{aligned} \right\} \Rightarrow   \big|(a_r^2 +  b_i^2) + (a_i^2 +b_r^2)\big| \ge 2\big| a_rb_i+a_ib_r \big|.
\end{aligned}
\end{equation}
Then, it is easy to see that $|a|^2 + |b|^2$ can be  lower bounded by
\vspace{-0.2cm}\begin{equation}\label{AMGM}
|a|^2 + |b|^2 = \frac{1}{2}\big|(a_r^2 +  b_r^2) + (a_i^2 +b_i^2)\big| + \frac{1}{2}\big|(a_r^2 + b_i^2) + (a_i^2 + b_r^2)  \big| \ge |a_rb_r - a_ib_i| + |a_rb_i + a_ib_r|.
\end{equation}
Besides, since
$
ab =  (a_ra_i - b_rb_i) + j(a_rb_i  + a_ib_r)
$ holds,
\eqref{AMGM} can be written  as
$|a|^2 + |b|^2 \ge \big|\Re(ab)\big| +  \big|\Im(ab)\big|$. This thus completes the proof.
\end{proof}
\end{appendix}

\vspace{-0.0cm}

\bibliography{IRS_time}
\bibliographystyle{IEEETran}

\end{document}